\documentclass[11pt]{article}
\usepackage[margin=1.3in]{geometry}

\usepackage[toc,page]{appendix}
\usepackage{graphicx}
\usepackage{hyperref}
\usepackage[T1]{fontenc}
\usepackage[latin1]{inputenc}
\usepackage{amsmath,amssymb,mathrsfs}
\usepackage{amsthm}
\usepackage{marginnote}
\usepackage{multicol}
\usepackage{multirow}
\usepackage{tikz}
\usepackage{enumitem}
\usepackage{subfig}
\usepackage{enumitem}
\setlist{topsep=1mm,noitemsep}
\usepackage{algorithm}
\usepackage{algpseudocode,float}

\algtext*{EndIf}

\makeatletter
\newenvironment{breakablealgorithm}
  {
   \begin{center}
     \refstepcounter{algorithm}
     \hrule height.8pt depth0pt \kern2pt
     \renewcommand{\caption}[2][\relax]{
       {\raggedright\textbf{\ALG@name~\thealgorithm} ##2\par}%
       \ifx\relax##1\relax 
         \addcontentsline{loa}{algorithm}{\protect\numberline{\thealgorithm}##2}%
       \else 
         \addcontentsline{loa}{algorithm}{\protect\numberline{\thealgorithm}##1}%
       \fi
       \kern2pt\hrule\kern2pt
     }
  }{
     \kern2pt\hrule\relax
   \end{center}
  }
\makeatother

\newcommand{\bbR}{\mathbb{R}}
\newcommand{\obj}{\textup{obj}}
\newcommand{\bpogp}{(BPO)$'$}
\newcommand{\drop}[2]{\textup{drop}_{#1}({#2})}

\newcommand{\pos}{\mathscr{P}}
\newcommand{\negat}{\mathscr{N}}
\newcommand{\posneg}{\mathscr{PN}}
\newcommand{\negpos}{\mathscr{NP}}
\newcommand{\ie}{i.e., }
\newcommand{\eg}{e.g., }
\renewcommand{\mid}{\, : \,}

\usepackage{color,soul}
\setulcolor{red}
\usepackage{schemata}
\usepackage{commath}

\newtheorem{theorem}{Theorem}
\newtheorem{proposition}{Proposition}

\newtheorem{claim}{Claim}

\newtheorem{example}{Example}

\theoremstyle{remark}
\newtheorem{remark}{Remark}
\newenvironment{cpf}[1][]
{\begin{trivlist} \item[] {\em #1.}}
{$\hfill\diamond$ \end{trivlist}}

\begin{document}

\title{On the complexity of binary polynomial optimization over acyclic hypergraphs}

\author{Alberto Del~Pia
\thanks{Department of Industrial and Systems Engineering \& Wisconsin Institute for Discovery,
             University of Wisconsin-Madison, Madison, WI, USA.
             E-mail: {\tt delpia@wisc.edu}.
             }
\and
Silvia Di~Gregorio
\thanks{Faculty of Computer Science, TU Dresden, N\"othnitzer Stra{\ss}e 46, Dresden, 01187, Germany.
             E-mail: {\tt silvia.di\_gregorio@tu-dresden.de}.
             }}

\date{December 14, 2022}

\maketitle

\begin{abstract} 
In this work, we advance the understanding of the fundamental limits of computation for Binary Polynomial Optimization (BPO), which is the problem of maximizing a given polynomial function over all binary points. 
In our main result we provide a novel class of BPO that can be solved efficiently both from a theoretical and computational perspective.
In fact, we give a strongly polynomial-time algorithm for 
instances whose corresponding hypergraph is $\beta$-acyclic.
We note that the $\beta$-acyclicity assumption is natural in several applications including relational database schemes and the lifted multicut problem on trees.
Due to the novelty of our proving technique, we obtain an algorithm which is interesting also from a practical viewpoint.
This is because our algorithm is very simple to implement and the running time is a polynomial of very low degree in the number of nodes and edges of the hypergraph.
Our result completely settles the computational complexity of BPO over acyclic hypergraphs, since the problem is NP-hard on $\alpha$-acyclic instances.
Our algorithm can also be applied to any general BPO problem that contains $\beta$-cycles.
For these problems, the algorithm returns a smaller instance together with a rule to extend any optimal solution of the smaller instance to an optimal solution of the original instance.
\end{abstract}


\section{Introduction}

In \emph{binary polynomial optimization} we seek a binary point that maximizes a given polynomial function.
This fundamental problem has a broad range of applications in several areas, including operations research, engineering, computer science, physics, biology, finance, and economics
(see \eg \cite{BorHam02,KocHaoGloLewLuWanWan14,dPKhaSah19}).

In order to formalize this optimization problem, a hypergraph representation is often used \cite{dPKha17}.
A \emph{hypergraph} $G$ is a pair $(V, E)$, where $V$ is the \emph{node set} and $E$ is the \emph{edge set}, which is a family of non-empty subsets of $V$.
We remark that the edge set $E$ may contain parallel edges and loops, as opposed to the setting considered in \cite{dPKha17,dPKha18MPA,dPKha21}. 
In the hypergraph representation, each node represents a variable of the given polynomial function, whereas every edge represents a monomial.
Therefore, any binary polynomial optimization problem can be formulated as
\begin{equation}\label{pb}
\tag{BPO}
\begin{aligned}
\max & \quad \sum_{v \in V} p_v x_v + \sum_{e \in E} p_e \prod_{v \in e} x_v \\
\text{s.t. } &\quad x \in \{0,1\}^V.    
\end{aligned}
\end{equation}
In this formulation, $x$ is the decision vector, and an instance comprises of a hypergraph $G = (V,E)$ together with a \emph{profit} vector $p \in \mathbb{Z}^{V \cup E}$.
We remark that a rational profit vector can be scaled to be integral by multiplying it by the least common multiple of the denominators and this transformation leads to a polynomial growth of the instance size (see Remark 1.1 in~\cite{ConCorZamBook}).

The main goal of this paper is that of advancing the understanding of the fundamental limits of computation for \eqref{pb}.
In fact, while there are several known classes of binary \emph{quadratic} optimization that are polynomially solvable (see for instance \cite{Bar86,Pad89,Cra93,Lau09,Mic21}), very few classes of higher degree \eqref{pb} are known to be solvable in polynomial-time.
These are instances that have:
(i) incidence graph or co-occurrence graph of fixed treewidth \cite{CraHanJau90,Lau09,BieMun18}, or
(ii) objective function whose restriction to $\{0, 1\}^n$ is supermodular (see Chapter 45 in \cite{SchBookCO}),
or
(iii) a highly acyclic structure \cite{dPKha21}, which we discuss in detail below.

Notice that, in the quadratic setting, the hypergraphs representing the instances are actually graphs.
It is known that instances over acyclic graphs can be solved in strongly polynomial time \cite{CraHanJau90}.
Motivated by this fact, it is natural to analyze the computational complexity of \eqref{pb} in the setting in which the hypergraph $G$ does not contain any cycle.
However, for hypergraphs, the definition of cycle is not unique.
As a matter of fact, one can define Berge-cycles, $\gamma$-cycles, $\beta$-cycles, and $\alpha$-cycles \cite{Fag83,JegNdo09}.
Correspondingly, one obtains Berge-acyclic, $\gamma$-acyclic, $\beta$-acyclic, and $\alpha$-acyclic hypergraphs, in increasing order of generality.
The definitions of $\beta$-acyclic and $\alpha$-acyclic hypergraph are given in Sections~\ref{sec: contrib} and~\ref{sec 1.2}, and we refer the reader to \cite{Fag83} for the remaining definitions.

In \cite{dPKha21}, Del Pia and Khajavirad show that it is possible to solve \eqref{pb} in polynomial-time if the corresponding hypergraph is kite-free $\beta$-acyclic.
It should be noted that this class of hypergraphs lies between $\gamma$-acyclic and $\beta$-acyclic hypergraphs.
This result is obtained via \emph{linearization,} which is a technique that consists in linearizing the polynomial objective function via the introduction of new variables $y_e$, for $e \in E$, defined by $y_e = \prod_{v \in e} x_v$.
This leads to an extended integer linear programming formulation in the space of the $(x,y)$ variables, which is obtained by replacing the nonlinear constraints $y_e = \prod_{v \in e} x_v$, for all $e \in E$, with the inequalities that describe its convex hull on the unit hypercube \cite{For}.
The convex hull of the feasible points is known as the \emph{multilinear polytope}, as defined in \cite{dPKha17}.
The tractability result in \cite{dPKha21} is then achieved by providing a linear programming extended formulation of the multilinear polytope of polynomial size.
The linearization technique also led to several other polyhedral results for \eqref{pb}, including \cite{dPKha17,CraRod17,dPKha18MPA,dPKha18SIOPT,BucCraRod18,dPKha21,dPDiG21,HojPfeWal19}.

A different approach to study binary polynomial optimization involves quadratization techniques \cite{Ros75,FreDri05,BucRin07,Ish09,Ish11}.
The common idea in the quadratization approaches is to add additional variables and constraints so that the original polynomial can be expressed in a higher dimensional space as a new quadratic polynomial.
The reason behind it is that, in this way, it is possible to exploit the vast literature available for the quadratic case.
An alternative approach is to use a different formalism altogether like pseudo-Boolean optimization \cite{HamRosRud63,HamRud68,CraHanJau90,BorHam02,BorGru12,BorCraRod20}.
Pseudo-Boolean optimization is a more general framework, as in fact the goal is to optimize set functions that admit closed algebraic expressions.

\subsection{A strongly polynomial-time algorithm for $\beta$-acyclic hypergraphs}
\label{sec: contrib}

Our main result is an algorithm that solves \eqref{pb} in strongly polynomial-time whenever the  hypergraph corresponding to the instance is $\beta$-acyclic.
To formally state our tractability result, we first provide the definition of $\beta$-acyclic hypergraph \cite{Fag83}.
A hypergraph is \emph{$\beta$-acyclic} if it does not contain any $\beta$-cycle. 
A \emph{$\beta$-cycle} of length $q$, for some $q \geq 3$, is a sequence $v_1$, $e_1$, $v_2$, $e_2$, $\dots$, $v_q$, $e_q$, $v_1$ such that $v_1$, $v_2$, $\dots$, $v_q$ are distinct nodes, $e_1$, $e_2$, $\dots$, $e_q$ are distinct edges, and $v_i$ belongs to $e_{i-1}, e_i$ and no other $e_j$ for all $i = 1,\dots,q$, where $e_0 = e_q$.

Our algorithm is based on a dynamic programming-type recursion.
The idea behind it is to successively remove a nest point from $G$, until there is only one node left in the hypergraph. 
In fact, we observe that optimizing the problem becomes trivial when there is only one node left.
A node $u$ of a hypergraph is a \emph{nest point} if for every two edges $e,f$ containing $u$, either $e \subseteq f$ or $f \subseteq e$.
Equivalently, the set of the edges containing $u$ is totally ordered.
Observe that, in connected graphs with at least two nodes, nest points coincide with leaves.
Therefore, nest points can be seen as an extension of the concept of leaf in a graph to the hypergraph setting.
Before going forward, we remark that finding a nest point in a hypergraph can be done in strongly polynomial-time by brute force \cite{OrdPauSze13}. 
We denote by $\tau$ the number of operations required to find a nest point, which is bounded by a polynomial in $\abs{V}$ and $\abs{E}$. 
We are now ready to state our main result.

\begin{theorem} \label{theobeta}
There is a strongly polynomial-time algorithm to solve \eqref{pb}, provided that the input hypergraph $G = (V, E)$ is $\beta$-acyclic.
In particular, the number of arithmetic operations performed is $O(\abs{V}(\tau + \abs{E} + \abs{V}\log\abs{E}))$.  
\end{theorem}

The description of the algorithm and the proof of Theorem~\ref{theobeta} can be found in Section~\ref{sec: hypergraph}.
Theorem~\ref{theobeta} provides a novel class of \eqref{pb} that can be solved efficiently both from a theoretical and computational perspective.
In fact, this class of problems is not contained in the classes (i), (ii), or (iii) for which a polynomial-time algorithm was already known. 
This can be seen because a laminar hypergraph $G=(V,E)$ with edges $e_1 \subseteq e_2 \subseteq \dots \subseteq e_m = V$ is $\beta$-acyclic and does not satisfy the assumptions in (i).
%
Furthermore, it is simple to see that there exist polynomials whose restriction to $\{0,1\}^n$ is not supermodular and the corresponding hypergraph is $\beta$-acyclic.
Finally, it is well-known that the class of $\beta$-acyclic hypergraphs significantly extends the class of kite-free $\beta$-acyclic hypergraphs.

The concept of $\beta$-acyclicity is not interesting only in a theoretical context.
To the contrary, this assumption is quite natural in several real world applications.
A thorough discussion of this topic is not in the scope of this paper, where instead we only mention a couple of examples.
In the study of relational database schemes, the $\beta$-acyclicity assumption is renowned to be advantageous \cite{Fag83Painless}.
In fact, a number of basic and desirable properties in database theory turn out to be equivalent to acyclicity. 
A second example is given by the lifted multicut problem on trees, where the problem can be equivalently formulated via binary polynomial optimization \cite{LanAnd21}. 
The goal of the lifted multicut problem is to partition a given graph in a way that minimizes the total cost associated with having different pairs of nodes in different components. 
This problem has been shown to be very useful in the field of computer vision, in particular when applied to image segmentation~\cite{Beietal17}, object tracking~\cite{Tangetal17}, and motion segmentation~\cite{Keu17}. 
Even when the underlying graph is a tree, the lifted multicut problem is NP-hard. However it can be solved in polynomial time when we focus on paths rather than on trees. 
It is simple to observe that this special case is formulated with a polynomial whose hypergraph is $\beta$-acyclic.
Lastly, we observe that these $\beta$-acyclic hypergraphs can exhibit kites, and therefore do not fit into the previous studies \cite{dPKha21}.

The interest of Theorem~\ref{theobeta} also lies in the novelty of the proving technique
with respect to the other recent results in the field previously mentioned. 
In particular, our algorithm does not rely on linear programming, extended formulations, polyhedral relaxations, or quadratization. 
This in turn leads to two key advantages.
First, our algorithm is very simple to implement.
Second, we obtain a \emph{strongly} polynomial time algorithm (as opposed to a \emph{weakly} polynomial time algorithm) and
the running time is a polynomial of very low degree in the number of nodes and edges of the hypergraph.
These two key points contribute to making our algorithm interesting also from a practical viewpoint.
Furthermore, we remark that it is possible to recognize efficiently when \eqref{pb} is represented by a $\beta$-acyclic hypergraph~\cite{Fag83}.

Theorem~\ref{theobeta} has important implications in polyhedral theory as well. 
In particular, it implies that one can optimize over the multilinear polytope for $\beta$-acyclic hypergraphs in strongly polynomial-time.
By the polynomial equivalence of separation and optimization (see, \eg \cite{ConCorZamBook}), for this class of hypergraphs, the separation problem over the multilinear polytope can be solved in polynomial-time.

We remark that our algorithm in Theorem~\ref{theobeta} can be applied also to hypergraphs that are not $\beta$-acyclic.
In this case, the algorithm does not return an optimal solution to the given instance.
However, it returns a smaller instance together with a rule to construct an optimal solution to the original instance, given an optimal solution to the smaller instance.
Therefore, our algorithm can be used as a reduction scheme to decrease the size of a given instance.
Via computational experiments, we generate random instances and study the magnitude of this decrease.
In particular, the results of our simulations show that the percentage of removed nodes is on average $50\%$ whenever the number of the edges is half the number of nodes.
We discuss in detail this topic in Section~\ref{sec: computation}.

\subsection{Settling the complexity of \eqref{pb} over acyclic hypergraphs}
\label{sec 1.2}

Theorem~\ref{theobeta} allows us to completely settle the computational complexity of \eqref{pb} over acyclic hypergraphs.
More specifically, it can be seen that two hardness results hold for \eqref{pb} when the input hypergraphs belong to the next class of acyclic hypergraphs, in increasing order of generality, that is the one of $\alpha$-acyclic hypergraphs.
Several equivalent definitions of $\alpha$-acyclic hypergraphs are known (see, \eg \cite{BeFaMaYa83,Fag83,Bra14}).
In the following, we will use the characterization stated in Theorem~\ref{charactalphaacyclic} below.
This characterization is based on the concept of removing nodes and edges from a hypergraph.
When we \emph{remove} a node $u$ from $G=(V,E)$ we are constructing a new hypergraph $G' = (V',E')$ with $V' = V \setminus \{ u \}$ and $E' = \{ e \setminus \{u\} \mid e \in E, \ e \neq \{u\} \}$.
Observe that when we remove a node we might be introducing loops and parallel edges in the hypergraph.
When we \emph{remove} an edge $f$ from $G=(V,E)$, we construct a new hypergraph $G' = (V,E')$, where $E' = E \setminus \{ f \}$.

\begin{theorem}[\cite{BeFaMaYa83}]
\label{charactalphaacyclic}
A hypergraph $G$ is $\alpha$-acyclic if and only if the empty hypergraph $(\emptyset, \emptyset)$ can be obtained by applying the following two operations repeatedly, in any order: 
\begin{enumerate}
\item
if a node $v$ belongs to at most one edge, then remove $v$;
\item
if an edge $e$ is contained in another edge $f$, then remove $e$.
\end{enumerate}
\end{theorem}


We claim that both Simple Max-Cut and Max-Cut can be formulated as special cases of \eqref{pb} where the hypergraphs representing the problems are $\alpha$-acyclic.
It is well-known that both these problems can be formulated as binary quadratic problems \cite{ConCorZamBook}.
Then, we define the corresponding instance of \eqref{pb} starting from the graph representing the instance of the binary quadratic problem.
Namely, we construct the hypergraph by adding to the graph one edge of weight zero that contains all the nodes. 
Theorem~\ref{charactalphaacyclic} implies that such hypergraph is $\alpha$-acyclic.
At this point, it can be seen that the corresponding instance of \eqref{pb} is equivalent to the original quadratic instance.
Therefore, the known hardness results of Simple Max-Cut and Max-Cut \cite{GarJohSto,TreSorSudWil} transfer to this setting, yielding the following hardness result.

\begin{theorem}
\label{NPhard-intro}
\eqref{pb} over $\alpha$-acyclic hypergraph is strongly NP-hard.
Furthermore, it is NP-hard to obtain an $r$-approximation for \eqref{pb}, with $r > \frac{16}{17} \approx 0.94$. 
\end{theorem}

The reduction just described shows that the statement of Theorem~\ref{NPhard-intro} holds even if the values of the objective function belong to a restricted subset.
The interested reader can find more details in Section~\ref{app: alpha}. 
Together, Theorem~\ref{theobeta} and Theorem~\ref{NPhard-intro} completely settle the computational complexity of binary polynomial optimization over acyclic hypergraphs.

\section{A strongly polynomial-time algorithm for $\beta$-acyclic hypergraphs}
\label{sec: hypergraph}

In this section we present the general algorithm for $\beta$-acyclic instances.
We start with a simple discussion to provide some intuition about why and how the algorithm works.
We are indebted to an anonymous reviewer for providing this simple interpretation.
In the following discussion, we denote by $\obj(x)$, for $x \in \bbR^V$, the objective function of \eqref{pb}, and we let $u \in V$.
Factoring out variable $x_u$ from the monomials in $\obj(x)$ that contain it, 
we write $\obj(x)$ in the form
$$
\obj(x) = x_u q(x') + r(x'),
$$
where $x' \in \bbR^{V\setminus \{u\}}$ is obtained from $x$ by dropping the component $x_u$, and where $q$ and $r$ are polynomials from $\bbR^{V\setminus \{u\}}$ to $\bbR$.
If $u$ is a nest point of the hypergraph $G$, the monomials in $q$ are totally ordered.
This special structure allows us to obtain efficiently a new polynomial $f$ from $\bbR^{V\setminus \{u\}}$ to $\bbR$ such that, for every $x' \in \{0,1\}^{V\setminus \{u\}}$, we have
$$
f(x') = 
\begin{cases}
q(x') & \text{if $q(x') > 0$} \\
0 & \text{if $q(x') \le 0$}.
\end{cases}
$$
The construction of the polynomial $f$ is nontrivial, and 
a large part of the next section will be devoted to obtaining its coefficients.
Assume now that we have an optimal solution ${x'}^*$ to the optimization problem, with one fewer variable, defined by
\begin{equation*}
\begin{aligned}
\max & \quad f(x') + r(x') \\
\text{s.t. } & \quad x' \in \{0,1\}^{V\setminus \{u\}}.    
\end{aligned}
\end{equation*}
Due to the property of the function $f(x')$, the vector $x^*$, obtained from ${x'}^*$ by adding component 
$$
x^*_u := 
\begin{cases}
1 & \text{if $q({x'}^*) > 0$} \\
0 & \text{if $q({x'}^*) \le 0$},
\end{cases}
$$
is an optimal solution to \eqref{pb}.

This idea is then used recursively to remove one variable at every iteration.
Since the hypergraph is $\beta$-acyclic, at each iteration there is a nest point, and so this recursion can be applied until only one variable remains.
At that point the problem can be solved trivially, and the construction of the optimal solution is performed in the reverse order.


\subsection{Description of the algorithm}

In this section we present the detailed description of our algorithm.
Our algorithm makes use of a characterization of $\beta$-acyclic hypergraphs, which is based on the concept of removing nest points from the hypergraph.
We remind the reader that the operation of removing a node is explained in Section~\ref{sec 1.2}.
We are now ready to state this characterization of $\beta$-acyclic hypergraphs.

\begin{theorem}[\cite{Dur12}]\label{charact}
A hypergraph $G$ is $\beta$-acyclic if and only if after removing nest points one by one we obtain the empty hypergraph $(\emptyset,\emptyset)$.
\end{theorem}

We observe that Theorem \ref{charact} does not depend on the particular choice of the nest point to be removed at each step.
Theorem \ref{charact} implies that, for our purposes, it suffices to understand how to
reduce an instance of the problem to one obtained by removing a nest point $u$.
In particular, realizing how to update the profit vector is essential.
Once we solve the instance of the new problem without $u$, we decide whether to set the variable corresponding to $u$ to zero or one depending on the values of the variables of the other nodes in the edges containing $u$, which are given by the solution of the smaller problem.

Before describing the algorithm, we explain some notation that will be used in this section.
Let $u \in V$ be a nest point contained in $k$ edges.
Without loss of generality, we can assume that these edges are $e_1$, $e_2$, \dots, $e_k$ and that $e_1 \subseteq e_2 \subseteq \dots \subseteq e_k$.
For simplicity of notation, we denote by $e_0$ the set $\{ u \}$ and by $p_{e_0}$ the profit $p_u$. 
Moreover, we clearly have $e_0 \subseteq e_1$.
We will divide the subcases to consider based on the sequence of the signs of 
$$
p_{e_0}, \quad p_{e_0} + p_{e_1}, \quad p_{e_0} + p_{e_1} + p_{e_2}, \quad \dots \ , \quad p_{e_0} + p_{e_1} + \dots + p_{e_k}.
$$
Note that the number of subcases can be exponential in the number of edges, however we find a compact formula for the optimality conditions, which in turn yields a compact way to construct the new profit vector $p'$ for the hypergraph $G' = (V',E')$ obtained by removing $u$ from $G$.
We say that there is a \emph{flip} in the sign sequence whenever the 
sign of the sequence changes. More precisely, a flip is \emph{positive} if the sign sequence goes from non-positive to positive and the previous non-zero value of the sequence is negative.
Similarly, we say that a flip is \emph{negative} if the sequence goes from non-negative to negative and the previous non-zero value of the sequence is positive.
We say that an edge $e_i$ \emph{corresponds to a flip} in the sign sequence, if there is a flip between $\sum_{j=0}^{i-1} p_{e_j}$ and $\sum_{j=0}^{i} p_{e_j}$.

In order to describe the several cases easily, in a compact way, we partition the indices $0, \dots, k$ into four sets $\pos$, $\negat$, $\negpos$, and $\posneg$.
The first two sets are defined by 
\begin{align*}
\pos := & \ \{ i \mid i = 1,\dots,k, \ e_i \text{ corresponds to a positive flip} \}   ,     \\
\negat := & \ \{ i \mid i= 1,\dots,k, \ e_i \text{ corresponds to a negative flip} \}   .
\end{align*}
If there is at least one flip, the sets $\negpos$, and $\posneg$ are defined  as follows:
\begin{align*}
\negpos := 
& \ \{ 0, \dots, s - 1 \mid \text{$s$ is the first flip and $s \in \pos$} \}       \\
& \ \cup \{ i \mid \exists \, \text{two consecutive flips  $s \in \negat$, $t \in \pos$ s.t. } s +1 \leq i \leq t-1 \} \\ 
& \ \cup \{ t+1, \dots, k \mid \text{if $t$ is the last flip and $t \in \negat$} \} ,      \\
\posneg := 
& \ \{ 0, \dots, s - 1 \mid \text{$s$ is the first flip and $s \in \negat$} \}       \\
& \ \cup \{ i \mid \exists \, \text{two consecutive flips  $s \in \pos$, $t \in \negat$ s.t. } s +1 \leq i \leq t-1 \} \\ 
& \ \cup \{ t+1, \dots, k \mid \text{if $t$ is the last flip and $t \in \pos$} \}.      
\end{align*}
Otherwise, if there is no flip, we define
\begin{align*}
\negpos := 
& \ \{ 0, \dots, k \mid \text{if $p_{e_0} \le 0$} \}    ,   \\
\posneg := 
& \ \{ 0, \dots, k \mid \text{if $p_{e_0} > 0$} \}    .
\end{align*}

\begin{remark}\label{ordine}
We observe that  the indices $\{0,1,\dots,k\}$ cycle between $\negpos$, $\pos$, $\posneg$, $\negat$ following this order.
In fact, if $i \in \pos$ then the following indices must be in $\posneg$ until we reach an index that belongs to $\negat$.
Similarly, if $i \in \negat$ the indices after $i$ must belong to $\negpos$ until there is an index in $\pos$.
Note that it can happen that there
is an index in $\pos$ and the next index is in $\negat$.
If this happens, then there are no indices in $\posneg$ between these two indices.
Similarly, it may happen that there is an index in $\negat$ followed immediately by an index in $\pos$.
Moreover, the index $0$ belongs to either $\negpos$ or $\posneg$.
$\hfill\diamond$
\end{remark}
\begin{example}
Let us give an example to clarify the meaning of 
the sets $\pos$, $\negat$, $\negpos$, and $\posneg$.
Consider a nest point $u$, contained in the edges $e_1$, $e_2$, $e_3$, $e_4$, $e_5$ such that $e_1 \subseteq e_2 \subseteq e_3 \subseteq e_4 \subseteq e_5$.
Assume that $p_{e_0} = 3$, $p_{e_1} = -3$, $p_{e_2} = 1$, $p_{e_3} = -2$, $p_{e_4} = 3$, $p_{e_5} = 2$.
We can check that $p_{e_0} = 3 > 0$, $p_{e_0} +$ $p_{e_1} = 0$, $p_{e_0} +$ $p_{e_1} +$ $p_{e_2} = 1 > 0$, $p_{e_0} + p_{e_1} + p_{e_2} + p_{e_3} = -1 < 0$,  $p_{e_0} + p_{e_1} + p_{e_2} + p_{e_3} + p_{e_4} = 2 > 0$ and finally $p_{e_0} +$ $p_{e_1} +$ $p_{e_2} +$ $p_{e_3} +$ $p_{e_4} +$ $p_{e_5} = 4 > 0$.
The indices $0,\dots,5$  are partitioned in the sets $\posneg = \{ 0, 1, 2, 5 \}$, $\negat = \{ 3 \}$, $\negpos = \emptyset$, $\pos = \{ 4 \}$.
Observe that here there are no indices in $\negpos$ when we go from the negative flip corresponding to $e_3$ to the next positive flip, which corresponds to $e_4$.
$\hfill\diamond$
\end{example}

Our algorithm acts differently whether all the edges containing the nest point $u$ are loops or not.
Let us now consider the case where $u$ is contained not only in loops.
In this case, for a vector $x' \in \{0,1\}^{V'}$, we define $\varphi (x') \in \{0,1\}$ that will assign the optimal value to the variable corresponding to the nest point $u$, given the values of the variables corresponding to the nodes in $V'$.
We denote by $\mu = \mu(x')$ the largest index $i \in \{ 0, \dots,k\}$, such that $x'_v = 1$ for every $v \in e_i \setminus \{ u \}$. 
Note that all the edges $e$ that are loops $\{ u \}$ satisfy trivially the condition $x'_v = 1$ for every $v \in e \setminus \{ u \}$, as $e \setminus \{ u \} = \emptyset$.
In particular, $e_0$ always satisfies this condition, hence $\mu$ is well defined. 
We then set
\begin{align*}
\varphi(x') := 
\begin{cases}
1 & \text{if $\mu \in \pos \cup \posneg$} \\
0 & \text{if $\mu \in \negat \cup \negpos$}.
\end{cases}
\end{align*}

\smallskip

In our algorithm we decide to keep loops and parallel edges for ease of exposition.
An additional reason is that we avoid checking for loops and parallel edges at every iteration.
Furthermore, in this way there is a bijection between $\{ e \in E \mid e \neq \{ u \} \}$ and $E'$, which will be useful in the arguments below.
In order to construct the new profit vector $p'$, it is convenient to give a name to the index of the first edge in $e_0 \subseteq e_1 \subseteq \dots \subseteq e_k$ that is not equal to $\{ u \}$.
We denote this index by $\lambda$.
We remark that when $u$ is not contained only in loops, the index $\lambda$ is well defined.
Next, observe that $p' \in \mathbb{R}^{V' \cup E'}$.
We will use an abuse of notation for the indices of $p'$ corresponding to the edges in $E'$ obtained from $e_{\lambda}$, $\dots$, $e_k$ by removing $u$.
We denote these indices by $e_{\lambda}$, $\dots$, $e_k$, even if these edges belong to $E$.
This abuse of notation does not introduce ambiguity because of the bijection between $\{ e \in E \mid e \neq \{ u \} \}$ and $E'$ and the fact that $\{ e_i \in E \mid i = \lambda, \dots, k \} \subseteq \{ e \in E \mid e \neq \{ u \} \}$. 
We are now ready to present our algorithm for $\beta$-acyclic hypergraphs, which we denote by \texttt{Acyclic($G,p$)}.

\bigskip

\begin{breakablealgorithm}
	\caption{\texttt{Acyclic($G,p$)} }
	\label{algbeta}
 	\begin{algorithmic}[1]
	\State Find a nest point $u$. Let $e_1 \subseteq e_2 \subseteq \dots \subseteq e_k$ be the edges containing it.
	\State Compute $\pos$, $\negat$, $\negpos$, $\posneg$.
	\State Construct the hypergraph $G' = (V',E')$ by removing $u$ from $G$. 
	
	\If{$e_k = \{ u\}$}
		\State Set $x^*_u := 
				\begin{cases}
				1 & \text{if $\sum_{i=0}^k p_{e_i} \geq 0$} \\
				0 & \text{if $\sum_{i=0}^k p_{e_i} < 0$} \ . 
				\end{cases}$
		
		\If{$\abs{V} > 1$}		
			\State Set 
			$p'_t := p_t$ for all $t \in V' \cup E'$.
			\State Set $x' := $ \texttt{Acyclic($G',p'$)}.
			\State Set $x^*_w := x'_w$ for all $w \in V'$.
		\EndIf

	\Else
		\State Find $\lambda$.
		\State Set 
		$p'_t := p_t$ for all $t \in V' \cup E' \setminus \{e_1,\dots,e_k\}$.
		\State Set $p'_{e_i} := 
				\begin{cases}
				0 & \text{for every $i \in \negpos \cap \{ i \mid i \geq \lambda \}$} \\
				\sum_{r=0}^{i} p_{e_r} & \text{for every $i \in \pos \cap \{ i \mid i \geq \lambda \}$}  \\ 
				p_{e_i} & \text{for every $i \in \posneg \cap \{ i \mid i \geq \lambda \}$}  \\ 
				- \sum_{r=0}^{i-1} p_{e_r} & \text{for every $i \in \negat \cap \{ i \mid i \geq \lambda \}$}  \ .
				\end{cases}$

		\State Set $x' := $ \texttt{Acyclic($G',p'$)}.
		\State Set $x^*_w := x'_w$ for all $w \in V'$.
		\State Set $x^*_u := \varphi(x')$.
	\EndIf
	\State \Return $x^*$.
	\end{algorithmic}
\end{breakablealgorithm}

\bigskip

The remainder of the section is organized as follows: In Section~\ref{sec: example}, we present an example of the execution of the algorithm; In Section~\ref{sec: correctness}, we show the correctness of the algorithm; In Section~\ref{sec: runtime}, we provide the analysis of the running time.

\subsection{Example of the execution of the algorithm}
\label{sec: example}


In this section, we show how \texttt{Acyclic($G,p$)} works by running it on an example.
We choose a $\beta$-acyclic hypergraph $G$ that is not kite-free (see \cite{dPKha21}), since no other polynomial-time algorithm is known for instances of this type.
We use the notation defined so far in this section. 
The input hypergraph $G$, together with the hypergraphs produced by the algorithm throughout its execution, is represented in the Figure~\ref{fig}.
The profits of the edges are written next to the label of the corresponding edge.
Labels are always outside their edges.
Moreover, we denote by $\lambda^{(i)}$, $\mu^{(i)}$, $\varphi^{(i)}$ the values of $\lambda$, $\mu(x')$, $\varphi(x')$ in the $i$-th step of the algorithm.
Similarly we call $G^{(i)}$ the hypergraph that is built at the end of the $i$-th step.
For ease of exposition, we will keep the same names for the edges throughout the execution of the algorithm.
For example, in $G^{(3)}$ we do not change the names of $e_3$ and $e_4$ to $e_1$ and $e_2$ respectively.

\begin{figure}[!htbp]
\label{fig}
\centering
\includegraphics[width=\textwidth]{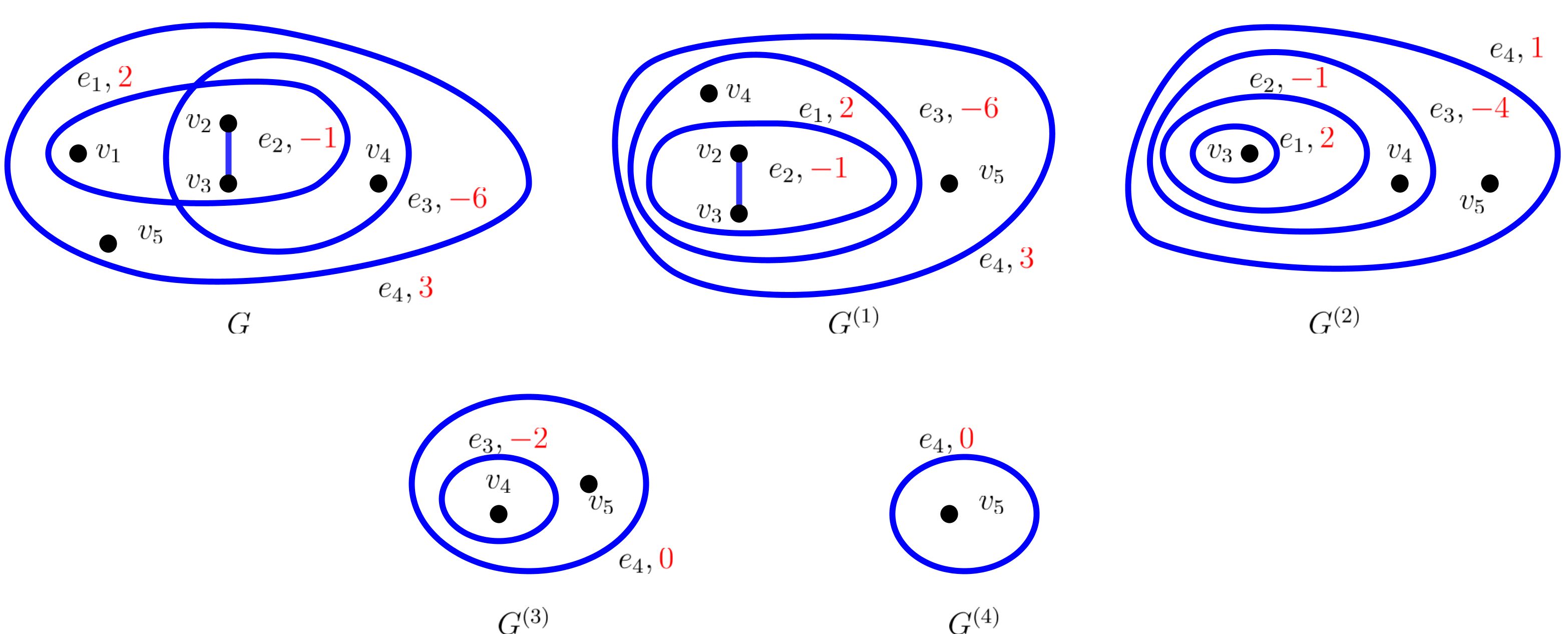}
\caption{The input hypergraph, and the hypergraphs produced at each iteration by \texttt{Acyclic($G,p$)}.}\label{fig}
\end{figure}

We define the profit vector as follows: $p_{v_1} = 1$, $p_{v_2} = 3$, $p_{v_3} = 2$, $p_{v_4} = -1$, $p_{v_5} = 1$, $p_{e_1} = 2$, $p_{e_2} = -1$, $p_{e_3} = -6$, $p_{e_4} = 3$.

\smallskip
\noindent \textbf{Iteration 1:}
Observe that $v_1$ is a nest point of $G$, as it belongs only to $e_1$, $e_4$, and these edges are such that $e_1 \subseteq e_4$.
Moreover, $\lambda^{(1)} = 1$. 
First of all we need to compute the sets $\pos$, $\negat$, $\negpos$, $\posneg$.
In order to do so, observe that $p_{v_1}$ is non-negative, as well as $p_{v_1} + p_{e_1}$ and $p_{v_1} + p_{e_1} + p_{e_4}$.
This means that $\posneg = \{0,1,4\}$. 
The node $v_1$ is removed from $G$, and the result is the hypergraph $G^{(1)}$ showed in Figure \ref{fig}.
We denote the profits corresponding to $G^{(1)}$ by $p^{(1)}$.
Therefore, at this step $p^{(1)}_v := p_v$ for every $v \in \{v_2, v_3, v_4, v_5\}$ and $p^{(1)}_e := p_e$ for every $e \in \{e_1, e_2, e_3, e_4\}$.

\smallskip
\noindent \textbf{Iteration 2:}
Here \texttt{Acyclic($G^{(1)},p^{(1)}$)} removes $v_2$.
We remark that $v_2$ is now a nest point for $G^{(1)}$, even if not for $G$.
In fact, in $G^{(1)}$ we have $e_1 \subseteq e_2 \subseteq e_3 \subseteq e_4$.
Furthermore, $\lambda^{(2)} = 1$. 
Here, $p^{(1)}_{v_2}$, $p^{(1)}_{v_2} + p^{(1)}_{e_1}$, $p^{(1)}_{v_2} + p^{(1)}_{e_1} + p^{(1)}_{e_2}$,  are non-negative, while $p^{(1)}_{v_2} + p^{(1)}_{e_1} + p^{(1)}_{e_2} + p^{(1)}_{e_3}$ is negative, and $p^{(1)}_{v_2} + p^{(1)}_{e_1} + p^{(1)}_{e_2} + p^{(1)}_{e_3} + p^{(1)}_{e_4}$ is positive.
Thus, $\posneg = \{0,1,2\}$, $\negat = \{3\}$, $\pos = \{4\}$, $\posneg = \emptyset$.
Next, we construct $G^{(2)}$.
We define the profits $p^{(2)}$ as follows: $p^{(2)}_{e_1} := p^{(1)}_{e_1} = 2$, and $p^{(2)}_{e_2} := p^{(1)}_{e_2} = -1$, however we define $p^{(2)}_{e_3} := - p^{(1)}_{v_2} - p^{(1)}_{e_1} - p^{(1)}_{e_2} = -4$ and $p^{(2)}_{e_4} := p^{(1)}_{v_2} + p^{(1)}_{e_1} + p^{(1)}_{e_2} + p^{(1)}_{e_3} + p^{(1)}_{e_4} = 1$.
Moreover, $p^{(2)}_{v} := p^{(1)}_{v}$ for every node $v$ of $G^{(2)}$, this means that $p^{(2)}_{v_3} := 2$, $p^{(2)}_{v_4} := -1$, $p^{(2)}_{v_5} := 1$.

\smallskip
\noindent \textbf{Iteration 3:}
Now it's the turn of $v_3$, which is a nest point of $G^{(2)}$.
Here $\lambda^{(3)} = 3$. 
It is easy to check that all sums $p^{(2)}_{v_3}$, $p^{(2)}_{v_3} + p^{(2)}_{e_1}$, $p^{(2)}_{v_3} + p^{(2)}_{e_1} + p^{(2)}_{e_2}$ are positive, whereas $p^{(2)}_{v_3} + p^{(2)}_{e_1} + p^{(2)}_{e_2} + p^{(2)}_{e_3}$ is negative, and $p^{(2)}_{v_3} + p^{(2)}_{e_1} + p^{(2)}_{e_2} + p^{(2)}_{e_3} + p^{(2)}_{e_4}$ is equal to zero.
Therefore $\posneg = \{0, 1,2\}$, $\negat = \{3\}$, $\negpos = \{4\}$, $\pos = \emptyset$.
The hypergraph $G^{(3)}$ is constructed by removing $v_3$ from $G^{(2)}$.
Observe that, as we remove $v_3$, we are also removing $e_1$ and $e_2$, since $e_1 = e_2 = \{ v_3 \}$.
Then, we set $p^{(3)}_{e_3} := -p^{(2)}_u - p^{(2)}_{e_1} - p^{(2)}_{e_2} = -3$, $p^{(3)}_{e_4} := 0$.
Finally, $p^{(3)}_{v_4} := -1$, $p^{(3)}_{v_5} := 1$.
We then iterate on the smaller hypergraph.

\smallskip
\noindent \textbf{Iteration 4:}
Next, $v_4$ is a nest point of $G^{(3)}$.
Observe that now we have that $\lambda^{(4)} = 4$. 
Here, $p^{(3)}_{v_4}$, $p^{(3)}_{v_4} + p^{(3)}_{e_3}$, $p^{(3)}_{v_4} + p^{(3)}_{e_3} + p^{(3)}_{e_4}$ are all negative.
Hence, $\negpos = \{0,3,4\}$. 
We construct $G^{(4)}$.
It is easy to see that $p^{(4)}_{e_4}$ is set equal to zero.
Therefore, we define 
$p^{(4)}_{v_5} := p^{(3)}_{v_5} = 1$.

\smallskip
\noindent \textbf{Iteration 5:}
We have arrived at the last step of the algorithm.
Indeed, $v_5$ is the only node in $G^{(4)}$.
Observe that it is useless to compute $\pos$, $\negat$, $\negpos$, $\posneg$, and $G^{(5)}$ in this last iteration.
So, we skip it.
We introduce $e_0 = \{v_5\}$ and let $p^{(4)}_{v_5} := p^{(4)}_{v_5}$, $p^{(4)}_{v_5} := 0$.
We check that $p^{(4)}_{v_5} + p^{(4)}_{e_4} = p^{(4)}_{v_5} > 0$.
So, we set $x_{v_5} := 1$.

\smallskip
At this point, we are ready to compute $x_{v_1}$, $x_{v_2}$, $x_{v_3}$, and $x_{v_4}$.
We start from computing $x_{v_4}$.
Since $x_{v_5} = 1$ and $e_4$ is the only edge in $G^{(4)}$, it follows that $\mu^{(4)} = 4$. 
Recall that $4 \in \negpos$ in iteration number 4.
Then, by the definition of $\varphi^{(4)}$, 
we set $x_{v_4} := 0$.
Now we look at $x_{v_3}$.
In this case $\mu^{(3)} = 2$. 
This follows from the facts that we have just set $x_{v_4} = 0$ and that $v_3$ belongs to all the edges of $G^{(2)}$.
Since $2 \in \posneg$ in iteration number 3, we set $x_{v_3} := 1$.
Next, consider $x_{v_2}$.
Similarly to before, $\mu^{(2)} = 2$, 
since $x_{v_4} = 0$.
Again, we have that $2 \in \posneg$ in iteration number 2.
Hence, we define $x_{v_2} := 1$.
It remains to compute $x_{v_1}$.
In order to compute $\mu^{(1)}$, we need to consider the edges containing $v_1$ in $G$, which are $e_1$ and $e_4$.
Since $x_{v_4} = 0$, we 
find that $\mu^{(1)} = 1$. 
Therefore we set $x_{v_1} := 1$, since $1 \in \posneg$ in the first iteration. 
Then, an optimal solution of the problem is $x = (x_{v_1}, x_{v_2}, x_{v_3}, x_{v_4}, x_{v_5}) = (1, 1, 1, 0, 1)$.


\subsection{Correctness of the algorithm}
\label{sec: correctness}

In this section, we show that \texttt{Acyclic($G,p$)} is correct.

\begin{proposition} \label{prop correct}
The algorithm \texttt{Acyclic($G,p$)} returns an optimal solution to \eqref{pb}, provided that $G$ is $\beta$-acyclic.
\end{proposition}

\begin{proof}
We prove this proposition by induction on the number of nodes.
We start from the base case, that is when $\abs{V} = 1$. 
It follows that $e = \{u\}$ for all $e \in E$, since $\{ e_1, \dots, e_k\} = E$.
In this case the algorithm only performs lines 1-5 and line 17. 
There are only two possible solutions: either $x^*_u = 0$, or $x^*_u = 1$.
The algorithm computes the objective corresponding to $x^*_u = 1$.
If the objective is non-negative, it sets $x^*_u := 1$, otherwise it sets $x^*_u := 0$.
The solution provided by the algorithm is optimal, since we are maximizing.

Next we consider the inductive step, and analyze the correctness of \texttt{Acyclic($G,p$)} when it removes a nest point.
We define $\obj(\cdot)$ to be the objective value of \eqref{pb} yielded by a binary vector in $\{0,1\}^{V}$.
Let $u$ be the nest point to be removed at a given iteration of the algorithm.
We denote by \bpogp\ the problem of the form \eqref{pb} over the hypergraph $G'$ and the profits $p'$, defined by 
\texttt{Acyclic($G,p$)}.
Likewise, let $\obj'(\cdot)$ be the objective value of \bpogp\ provided by a vector in $\{0,1\}^{V'}$. 
By the inductive hypothesis, the vector $x'$ defined in line 8 or 
14 
is optimal to \bpogp.
Our goal is to show that the returned solution $x^*$ is optimal to \eqref{pb}.

\smallskip

We consider first the case in which $e_k = \{u\}$, \ie when all the edges that contain $u$ are loops.
This implies that every edge $e \in E$ is either a loop $\{ u \}$ or does not contain the node $u$.
Therefore, an optimal solution to \eqref{pb} is obtained by combining an optimal solution to \bpogp {} with an optimal solution to the problem represented by the hypergraph $(\{ u \}, \{ e_1, \dots, e_k \})$ with profits $p_u$ and $p_{e_i}$, for $i = 1,\dots,k$.
By using the same proof of the base case, we can see that line 5 provide the optimal value of $x^*_u$.
Since the vector $x'$ is optimal to \bpogp, we can conclude that the vector $x^*$ returned by the algorithm is optimal.

\smallskip

Next, we consider the case in which $e_k$ is not a loop.
For notational simplicity, we introduce extensions of the functions $\mu$ and $\varphi$ with domain $\{0,1\}^V$ rather than $\{0,1\}^{V'}$.
To do so, given a vector $x \in \{0,1\}^V$, we denote by $\drop{u}{x}$ the vector in $\{0,1\}^{V'}$ obtained from $x$ by dropping its entry corresponding to the node~$u$.
We then define $\mu(x) := \mu(\drop{u}{x})$ and $\varphi(x) := \varphi(\drop{u}{x})$.


\begin{claim}\label{claim: opt}
There exists an optimal solution $\tilde x$ to \eqref{pb} such that $\tilde x_u~=~\varphi(\tilde x)$.
\end{claim}

\begin{cpf}[Proof of Claim~\ref{claim: opt}]
To show this, let $\bar x$ be an optimal solution to \eqref{pb}.
If $\bar x_u = \varphi(\bar x)$, then we are done.
Thus, assume that $\bar x_u = 1 - \varphi(\bar x)$, and let $\tilde x$ be obtained from $\bar x$ by setting $\tilde x_u := \varphi(\bar x)$. 
Note however that $\varphi(\bar x) = \varphi(\tilde x)$, since $\bar x_v = \tilde x_v$ for all nodes $v \neq u$.
Therefore we want to show that $\tilde x$ is optimal.
The proof splits in two cases: either $\varphi(\tilde x) = 0$, or $\varphi(\tilde x) = 1$.

Consider the first case $\varphi(\tilde x) = 0$. 
Hence $\bar x_u = 1$ and $\tilde x_u = 0$.
Therefore, it follows that $\obj(\bar x) = \obj(\tilde x) + \sum_{i=0}^{\mu} p_{e_i}$. 
By definition of $\varphi$, we have $\mu = \mu(\tilde x) \in \negat \cup \negpos$, thus $\sum_{i=0}^{\mu} p_{e_i} \le 0$.
Then, we obtain that $\obj(\bar x) \le \obj(\tilde x)$ and $\tilde x$ is optimal to \eqref{pb} as well.

Assume now that we are in the second subcase, i.e., $\varphi(\tilde x) = 1$.
Therefore we have $\bar x_u = 0$, $\tilde x_u = 1$, and $\obj(\tilde x) = \obj(\bar x) +  \sum_{i=0}^{\mu} p_{e_i}$. 
Since $\mu \in \pos \cup \posneg$, it follows that $\sum_{i=0}^{\mu} p_{e_i} \ge 0$, therefore $\obj(\tilde x) \ge \obj(\bar x)$.
Thus, we can conclude that also $\tilde x$ is optimal to~\eqref{pb}.
\end{cpf}

We remark that, since $e_k \neq \{ u \}$, the index $\lambda$ is well defined and $\lambda \geq 1$.
From now on let $x$ be any vector $\{0,1\}^V$ such that $x_u = \varphi(x)$.
Let $\mu = \mu(x)$.

Our next main goal is to show the equality
\begin{equation}\label{eq:claim}
\obj(x) = 
\begin{cases}
\obj'(\drop{u}{x}), & \text{if $\lambda \in \negpos \cup \pos$} \\
\obj'(\drop{u}{x}) + \sum_{i=0}^{\lambda - 1} p_{e_i}, & \text{if $\lambda \in \posneg \cup \negat$} . 
\end{cases}
\end{equation}
We define the sets $A$ and $B$ as follows.
If $x_u = 0$ let $A := \emptyset$. 
Otherwise, that is if $x_u = 1$, we define $A := \{ 0, 1, \dots, \mu \}$.
In order to define $B$ we observe that either $\lambda \leq \mu$ or $\mu = \lambda -1$.
This is because $\lambda - 1$ is the index of the last loop $\{u\}$.
We then define $B := \{ \lambda, \dots, \mu \}$ if $\lambda \leq \mu$, otherwise we set $B := \emptyset$, if $\mu = \lambda -1$.
In order to prove \eqref{eq:claim}, it suffices to check that 
\begin{equation}\label{subclaim}
\sum_{i \in A} p_{e_i} = 
\begin{cases}
\sum_{i \in B} p'_{e_i} , & \text{if $\lambda \in \negpos \cup \pos$} \\
\sum_{i \in B} p'_{e_i} + \sum_{i=0}^{\lambda - 1} p_{e_i}, & \text{if $\lambda \in \posneg \cup \negat$} ,
\end{cases}
\end{equation}
by the definitions of $p'$ and $\drop{u}{x}$.
In the next claim, we study the value of $\sum_{i \in B} p'_{e_i}$, which is present in \eqref{subclaim}.

\begin{claim}\label{claim: sum B}
Let $\lambda \leq \mu$.
If $\pos \cap \{ \lambda, \dots, \mu \} = \emptyset$, then
\begin{equation}\label{subclaimB1}
\sum_{i \in B} p'_{e_i} = 
\begin{cases}
0, & \text{if $\lambda \in \negpos$} \\
\sum_{i=\lambda}^{\mu} p_{e_i} , & \text{if $\lambda \in \posneg$ and $\mu \in \posneg$} \\
- \sum_{i=0}^{\lambda - 1} p_{e_i}, & \text{if $\lambda \in \negat$ or if $\lambda \in \posneg$ and $\mu \in \negat \cup \negpos$} .
\end{cases}
\end{equation}
If $\pos \cap \{ \lambda, \dots, \mu \} \neq \emptyset$, then
\begin{equation}\label{subclaimB2}
\sum_{i \in B} p'_{e_i} = 
\begin{cases}
0, & \text{if $\lambda \in \negpos \cup \pos$ and $\mu \in \negat \cup \negpos$} \\
\sum_{i=0}^{\mu} p_{e_i} , & \text{if $\lambda \in \negpos \cup \pos$ and $\mu \in \pos \cup \posneg$} \\
- \sum_{i=0}^{\lambda - 1} p_{e_i}, & \text{if $\lambda \in \posneg \cup \negat$ and $\mu \in \negat \cup \negpos$} \\
\sum_{i=\lambda}^{\mu} p_{e_i} , & \text{if $\lambda \in \posneg \cup \negat$ and $\mu \in \pos \cup \posneg$} .
\end{cases}
\end{equation}
\end{claim}

\begin{cpf}[Proof of Claim~\ref{claim: sum B}]
Observe that $\sum_{i \in B} p'_{e_i}$ is not trivially equal to zero, since $\lambda \leq \mu$.

\smallskip

First, we assume that $\pos \cap \{ \lambda, \dots, \mu \} = \emptyset$.
In this case we can easily compute the value of $\sum_{i \in B} p'_{e_i}$.
Assume first that $\lambda \in \negpos$.
By Remark~\ref{ordine} it is easy to see that $\{\lambda, \dots, \mu \}$ must belong to $\negpos$.
Then, by definition of $p'$, it follows that $\sum_{i \in B} p'_{e_i} = 0$.
Next, consider the case in which $\lambda \in \posneg$ and $\mu \in \posneg$.
From Remark~\ref{ordine}, we can conclude that $\{\lambda, \dots, \mu\} \subseteq \posneg$.
By definition of $p'$, we can observe that $\sum_{i \in B} p'_{e_i} = \sum_{i=\lambda}^{\mu} p_{e_i}$.
Next, assume that $\lambda \in \negat$. 
Since $\lambda \in \negat$, it is easy to see that $\{\lambda + 1, \dots, \mu\} \subseteq \negpos$ by Remark~\ref{ordine}.
Hence by definition of $p'$, we get that $\sum_{i \in B} p'_{e_i} = p'_{e_{\lambda}} = - \sum_{i=0}^{\lambda - 1} p_{e_i}$.
Lastly, let $\lambda \in \posneg$ and $\mu \in \negat \cup \negpos$.
This implies that there must be exactly one index $q \in \negat \cap \{\lambda + 1, \dots, \mu\}$. 
By using the definition of $p'$ we obtain $\sum_{i \in B} p'_{e_i} = 
\sum_{i = \lambda}^{q -1} p'_{e_i} + p'_q + \sum_{i = q+1}^{\mu} p'_{e_i} = 
\sum_{i = \lambda}^{q -1} p_{e_i} - \sum_{i = 0}^{q-1} p_{e_i} = - \sum_{i=0}^{\lambda - 1} p_{e_i}$.
This ends the proof of~\eqref{subclaimB1}.

\smallskip

Next, we assume $\pos \cap \{ \lambda, \dots, \mu \} \neq \emptyset$.
We divide $\sum_{i \in B} p'_{e_i}$ in three parts. 
Let $\iota_1$ be the first index in $\pos \cap \{ \lambda, \dots, \mu \}$, and let $\iota_2$ be the last index in $\pos \cap \{ \lambda, \dots, \mu \}$. 
Note that it is possible that $\iota_1 = \iota_2$.
Then, we observe that 
\begin{equation}\label{brokensum}
\sum_{i \in B} p'_{e_i} = \sum_{i = \lambda}^{\iota_1 - 1} p'_{e_i} + \sum_{i = \iota_1}^{\iota_2 - 1} p'_{e_i} + \sum_{i = \iota_2}^{\mu} p'_{e_i} \ .
\end{equation}
Now we study the value of the sums in the right hand side of~\eqref{brokensum}.

We start by showing that
\begin{equation}\label{eq: zero}
\sum_{i = \iota_1}^{\iota_2 - 1} p'_{e_i} = 0 \ .
\end{equation}
If it is vacuous, then it is trivially equal to zero.
Then we assume that it is not vacuous.
Since $\iota_2 \in \pos$, the last index in this sum is in $\negat \cup \negpos$.
Because of the fact that the first index of the sum is in $\pos$ and by definition of $p'$, we can conclude that all the profits in $\sum_{i = \iota_1}^{\iota_2 - 1} p'_{e_i}$ cancel each other out.
Then, \eqref{eq: zero} holds.
From now on, in the analysis of \eqref{brokensum}, we will only focus on the values of $\sum_{i = \lambda}^{\iota_1 - 1} p'_{e_i}$ and $\sum_{i = \iota_2}^{\mu} p'_{e_i}$.

We consider the first of these two sums.
We prove that
\begin{equation}\label{eq: firstsum}
\sum_{i = \lambda}^{\iota_1 - 1} p'_{e_i} = 
\begin{cases}
0, & \text{if $\lambda \in \negpos \cup \pos$} \\
- \sum_{i=0}^{\lambda - 1} p_{e_i}, & \text{if $\lambda \in \posneg \cup \negat$} \ .
\end{cases}
\end{equation}
We start with analyzing the case in which $\lambda \in \negpos \cup \pos$.
If $\lambda \in \pos$, then $\iota_1 = \lambda$ and the sum is trivially equal to $0$.
Then we assume that $\lambda \in \negpos$.
Since $\iota_1$ is the first index in $\pos$ with $\iota_1 \ge \lambda$, Remark~\ref{ordine} implies that
all indices $\{\lambda, \dots, \iota_1 - 1 \}$ belong to $\negpos$.
Therefore, by definition of $p'$,
we conclude that $\sum_{i = \lambda}^{\iota_1 - 1} p'_{e_i} = 0$.
Next, let $\lambda \in \posneg \cup \negat$.
Here, the indices in $\{\lambda, \dots, \iota_1-1\}$ can be in $\posneg$, $\negat$, or $\negpos$.
Moreover, there must be exactly one index $q \in \negat \cap \{\lambda, \dots,\iota_1-1\}$.
Then, we can see that $\sum_{i = \lambda}^{\iota_1 - 1} p'_{e_i} = \sum_{i = \lambda}^{q -1} p'_{e_i} + p'_q + \sum_{i = q+1}^{\iota_1 - 1} p'_{e_i} = \sum_{i = \lambda}^{q -1} p_{e_i} - \sum_{i = 0}^{q-1} p_{e_i} = - \sum_{i=0}^{\lambda - 1} p_{e_i}$.
This concludes the proof of \eqref{eq: firstsum}.

It remains to compute $\sum_{i = \iota_2}^{\mu} p'_{e_i}$.
We show that
\begin{equation}\label{eq: secondsum}
\sum_{i = \iota_2}^{\mu} p'_{e_i} = 
\begin{cases}
0, & \text{if $\mu \in \negat \cup \negpos$} \\
\sum_{i=0}^{\mu} p_{e_i}, & \text{if $\mu \in \pos \cup \posneg$} \ .
\end{cases}
\end{equation}
Assume that $\mu \in \negat \cup \negpos$.
Since $\iota_2 \in \pos$, there must be an index $q \in \negat \cap \{ \iota_2 + 1, \dots, \mu \}$.
Then, $\sum_{i = \iota_2}^{\mu} p'_{e_i} = p'_{\iota_2} + \sum_{i=\iota_2+1}^{q-1} p'_{e_i} + p'_{e_q} + \sum_{i=q+1}^{\mu} p'_{e_i}$.
By using the definition of $p'$ we obtain the following: $\sum_{i = \iota_2}^{\mu} p'_{e_i} = \sum_{i=0}^{\iota_2} p_{e_i} + \sum_{i=\iota_2+1}^{q-1} p_{e_i} - \sum_{i=0}^{q-1} p_{e_i} = 0$.
We look at the second case, and we assume that $\mu \in \pos \cup \posneg$.
Then, by definition of $\iota_2$ and the fact that $\mu \in \pos \cup \posneg$, it follows that all indices in $\{ \iota_2+1, \dots, \mu \}$ belong to $\posneg$.
By the definition of $p'$, we can conclude that $\sum_{i = \iota_2}^{\mu} p'_{e_i} = p'_{\iota_2} + \sum_{i = \iota_2+1}^{\mu} p'_{e_i} = \sum_{i=0}^{\iota_2} p_{e_i} + \sum_{i=\iota_2+1}^{\mu} p_{e_i} = \sum_{i=0}^{\mu} p_{e_i}$.
This concludes the proof of \eqref{eq: secondsum}.

\smallskip

We conclude that \eqref{subclaimB2} holds, by combining appropriately the different cases of \eqref{eq: firstsum} and \eqref{eq: secondsum} into \eqref{brokensum}.
\end{cpf}

\smallskip

We are now ready to prove \eqref{eq:claim}.
This proof is divided in two cases, depending on the value of $x_u$.
The first case that we consider is when $x_u = \varphi(x) = 0$.
Therefore, we assume that $x_u = \varphi(x) = 0$.
As previously observed, we only need to show that \eqref{subclaim} holds.
hence $\sum_{i \in A} p_{e_i} = 0$.
Furthermore, $\varphi(x) = 0$ implies $\mu \in \negat \cup \negpos$.
We consider the two cases $\mu = \lambda - 1$ and $\lambda \le \mu$.
Consider the case $\mu = \lambda - 1$.
Then $\sum_{i \in B} p'_{e_i}$ is vacuous and equal to $0$.
Furthermore, if $\mu \in \negat \cup \negpos$ it means that $\lambda \in \negpos \cup \pos$. 
Hence \eqref{subclaim} holds. 
Next, we assume that $\lambda \le \mu$, which implies that $B$ is non-empty.
If $\pos \cap \{\lambda,\dots,\mu\} = \emptyset$, we get that $\sum_{i \in B} p'_{e_i} = 0$ if $\lambda \in \negpos$ by \eqref{subclaimB1}. 
Therefore \eqref{subclaim} is true.
Otherwise, if $\lambda \in \posneg \cup \negat$, then $\sum_{i \in B} p'_{e_i} = - \sum_{i=0}^{\lambda - 1} p_{e_i}$ since $\mu \in \negat \cup \negpos$. 
Hence \eqref{subclaim} holds also in this case. 
Therefore, we assume that $\pos \cap \{\lambda,\dots,\mu\} \neq \emptyset$.
We start from the case in which $\lambda \in \negpos \cup \pos$.
From~\eqref{subclaimB2}, we see that $\sum_{i \in B} p'_{e_i} = 0$, as $\mu \in \negat \cup \negpos$.
This concludes the proof of~\eqref{subclaim} when $\lambda \in \negpos \cup \pos$.
So now consider $\lambda \in \posneg \cup \negat$.
From \eqref{subclaimB2} we obtain $\sum_{i \in B} p'_{e_i} = - \sum_{i=0}^{\lambda - 1} p_{e_i}$ in this case.
Hence, we can conclude that \eqref{subclaim}
holds also if 
$\lambda \in \posneg \cup \negat$.

The remaining case to consider, in order to prove that \eqref{eq:claim} holds for every $x \in \{0,1\}^V$, is when $x_u = \varphi(x) = 1$.
Similarly to the previous case, we just need to show that \eqref{subclaim} holds.
Assume $x_u = \varphi(x) = 1$.
From $x_u = 1$ we obtain $\sum_{i \in A} p_{e_i} = \sum_{i=0}^{\mu} p_{e_i}$.
Because of $\varphi(x) = 1$, we know that $\mu \in \pos \cup \posneg$.
Once again, we consider the cases $\mu = \lambda - 1$ and $\lambda \le \mu$.
Assume $\mu = \lambda - 1$.
Then, we have that $B = \emptyset$ and $\sum_{i\in B} p'_{e_i}$ is equal $0$.
Moreover, we have $\lambda \in \posneg \cup \negat$, since $\mu \in \pos \cup \posneg$.
Then, it is easy to see that \eqref{subclaim} is true.
Next, we consider the case in which $\lambda \le \mu$.
We start from situation where $\lambda \in \negpos \cup \pos$.
By using Remark~\ref{ordine}, we observe that $\pos \cap \{\lambda,\dots,\mu\} \neq \emptyset$.
Then, we obtain that $\sum_{i \in B} p'_{e_i} = \sum_{i=0}^{\mu} p_{e_i}$ from \eqref{subclaimB2}.
Hence, \eqref{subclaim} holds. 
Assume now $\lambda \in \posneg \cup \negat$.
We first observe that it is possible that $\pos \cap \{\lambda,\dots,\mu\} = \emptyset$. 
This can happen only if $\lambda, \mu \in \posneg$.
In this case 
$\sum_{i \in B} p'_{e_i} = \sum_{i=\lambda}^{\mu} p_{e_i}$ by \eqref{subclaimB1}.
It is easy to see that \eqref{subclaim} holds in this case.
So assume instead that $\pos \cap \{\lambda,\dots,\mu\} \neq \emptyset$. 
Then, $\sum_{i \in B} p'_{e_i} = \sum_{i=\lambda}^{\mu} p_{e_i}$ by \eqref{subclaimB2}.
Therefore, \eqref{subclaim} is true.
This concludes the proof of~\eqref{eq:claim}.

\smallskip 

We are finally ready to show that the solution provided by the algorithm is optimal.
Let $\tilde x$ be an optimal solution to \eqref{pb} such that $\tilde x_u = \varphi(\tilde x)$.
We know that it exists by Claim~\ref{claim: opt}.
We denote by $x^*$ the solution returned by the algorithm, which is defined by
\begin{equation*}
x^*_w := \begin{cases}
x'_w,  & \text{if } w \neq u \\
\varphi(x'), \quad & \text{if } w = u \ .
\end{cases}
\end{equation*}
It is easy to see that $x^*_u = \varphi(x')$.
By the previous argument, it follows that \eqref{eq:claim} holds for both $\tilde x$ and $x^*$. 
Therefore, $\obj(x^*) = \obj'(x')$ and $\obj(\tilde x) = \obj'(\drop{u}{\tilde x})$, 
if $\lambda \in \negpos \cup \pos$.
Similarly, if $\lambda \in \posneg \cup \negat$, we obtain that $\obj(x^*) = \obj'(x') + \sum_{i=0}^{\lambda - 1} p_{e_i}$ and $\obj(\tilde x) = \obj'(\drop{u}{\tilde x}) + \sum_{i=0}^{\lambda - 1} p_{e_i}$. 
We are now ready to prove that $x^*$ is optimal to \eqref{pb}.
The optimality of $x'$ to \bpogp {} implies that $\obj'(x') \geq \obj'(\drop{u}{\tilde x})$.
This inequality implies $\obj(x^*) \geq \obj(\tilde x)$ in both cases.
Note that if $\lambda \in \posneg \cup \negat$ it suffices to add $\sum_{i=0}^{\lambda - 1} p_{e_i}$ on both sides of the inequality to see this.
Hence, we can conclude that $x^*$ is an optimal solution to \eqref{pb}.
\end{proof}

We remark that our algorithm is correct even if the profits are allowed to be real numbers.
However, for the purposes of the analysis of the algorithm, we chose to consider only the setting in which the profits are all integers.

\subsection{Analysis of the running time}
\label{sec: runtime}


In this section, we show that \texttt{Acyclic($G,p$)} runs in strongly polynomial time.
We remark that in this paper we use standard complexity notions in discrete optimization, and we refer the reader to the book \cite{SchBookIP} for a thorough introduction.
Our analysis is admittedly crude and provides a loose upper bound of the running time.
It could be further improved by paying particular attention to the data structure and to the exact number of operations performed in each step.
In our analysis, we choose to store the hypergraph $G=(V,E)$ by its node-edge 
incidence matrix.

The running time that we exhibit below is in terms of the time needed to find one nest point in $G$, which is denoted by $\tau$.
As mentioned in \cite{OrdPauSze13}, nest points can be found in polynomial-time by brute force.
Once we find one nest point, we also explicitly know the edges that contain it and their order under set inclusion.

\begin{proposition}\label{theoruntime}
The algorithm \texttt{Acyclic($G,p$)} is strongly polynomial, provided that $G=(V,E)$ is a $\beta$-acyclic hypergraph.
In particular, the number of arithmetic operations performed is $O(\abs{V}(\tau + \abs{E} + \abs{V}\log\abs{E}))$.
\end{proposition}

\begin{proof}
We first examine the number of arithmetic operations performed by the algorithm.
In line~1, there are at most $\tau$ operations to find a nest point $u$ and the ordered sequence of edges it belongs to, that is, $e_1 \subseteq e_2 \subseteq \dots \subseteq e_k$.
Line~2 requires $O(\abs{E})$ operations, between sums and comparisons, to compute the sets $\pos$, $\negat$, $\negpos$, $\posneg$.
In line~3, there are other $O(\abs{E})$ operations to remove $u$ from $G$ in order to construct the hypergraph $G'$, since it suffices to drop the $u$-th row from the incidence matrix.
We observe that we do not remove the columns of edges that might have become empty.
So, the incidence matrix could have some zero columns.
Line~4 takes $O(\abs{V})$ operations.
Next, there are $O(\abs{E})$ sums in the if condition in line~5.
Line~6 can be performed in constant time.
Then, line~7 requires $O(\abs{V}+\abs{E})$ operations, and line~9 takes $O(\abs{V})$ operations.
Next, finding $\lambda$ in line~11 requires $O(\abs{V}\log\abs{E})$ operations, by performing binary search on the ordered edges and checking the nodes they contain.
Consider now the construction of $p'$ in lines~12-13.
Line~12 takes $O(\abs{V}+\abs{E})$ operations. 
The profits $p'$ for the edges $e_{\lambda},\dots,e_k$ can be constructed with a total number of $O(\abs{E})$ operations.
Hence, constructing the smaller instance in both cases takes linear time.
It remains to consider the operations needed to construct $x^*$ from $x'$, see lines~15-16.
Line~15 requires $O(\abs{V})$ operations.
Now consider line~16.
Using the definition of the quantity $\varphi(x')$, it can be seen that 
the definition of $x^*_u$ requires $O(\abs{V}\log\abs{E})$ operations.
In fact it suffices to find $\mu(x')$.  

Therefore, each iteration of algorithm performs at most $\tau + O(\abs{E} + \abs{V}\log\abs{E})$ arithmetic operations. 
Moreover, we observe that \texttt{Acyclic($G,p$)} performs $\abs{V}$ iterations, thanks to Theorem~\ref{charact}.
We hence obtain that the total number of arithmetic operations performed by \texttt{Acyclic($G,p$)} is 
$O(\abs{V}(\tau + \abs{E} + \abs{V}\log\abs{E}))$.

To prove that the algorithm \texttt{Acyclic($G,p$)} is strongly polynomial, it remains to show that any integer produced in the course of the execution of the algorithm has size bounded by a polynomial in $\abs{V}+\abs{E}+\log U$, where $U$ is the largest absolute value of the profits in the instance (see page 362 in \cite{BerTsiBook}).
The numbers that are produced by the execution of the algorithm are the profits of the smaller instances. 
The only arithmetic operations involving the profits are addition and subtraction of the original profits.
In particular, this implies that the numbers produced are integers.
Moreover, only a polynomial number of operations $p(\abs{V},\abs{E})$ occur in the algorithm since its arithmetic running time is polynomial in $\abs{V}$ and $\abs{E}$.
Then, any integer obtained at the end of the algorithm must have absolute value less than or equal to $2^{p(\abs{V},\abs{E})}U$. 
Its bit size therefore is less than or equal to $p(\abs{V},\abs{E}) + \log U$.
\end{proof}

We close this section by observing that
the overarching structure of our algorithm, where nodes are removed one at a time, resembles that of the \emph{basic algorithm} for pseudo-Boolean optimization, which was first defined in the sixties~\cite{HamRosRud63,HamRud68}.
Except for this similarity, the two algorithms are entirely different.
For example, in the basic algorithm nodes can be removed in any order, but the running time can be exponential.
On the other hand, in our algorithm the node to be removed must be a nest point in order for the algorithm to be correct.
In particular, this allows us to define the updated profits and it is key in achieving a polynomial running time. 
In~\cite{CraHanJau90} the authors show that, if nodes are removed according to a ``$k$-perfect elimination scheme'', the basic algorithm runs in polynomial time for hypergraphs whose 
co-occurrence graph has fixed treewidth.
However, analyzing the laminar hypergraph discussed after the statement of Theorem~\ref{theobeta} in Section~\ref{sec: contrib}, it is simple to see that the basic algorithm does not run in polynomial time over $\beta$-acyclic hypergraphs, under \emph{any} choice of the node to be removed.

\section{Hardness for $\alpha$-acyclic hypergraphs}\label{app: alpha}

In this section, we describe the intractability results for \eqref{pb} over $\alpha$-acyclic instances, thereby showing Theorem~\ref{NPhard-intro}.
In order to prove these results, we will use polynomial reductions from Max-Cut and Simple Max-Cut to \eqref{pb}.
We recall that Max-Cut can be formulated as 
\begin{align*}
\max & \quad \sum_{\{u,v\} \in E} w_{\{u,v\}}(x_u + x_v - 2x_u x_v) \\
\text{s.t. } &\quad x \in \{0,1\}^V ,  
\end{align*}
where $G= (V,E)$ is the graph representing the instance of Max-Cut and $w \in \mathbb{Z}^E_+$ \cite{ConCorZamBook}.

Similarly to the $\beta$-acyclic case, we apply the idea of removing nodes and edges from a hypergraph.
Here, we will use it to show that the instances obtained via the polynomial reductions from Max-Cut and Simple Max-Cut to \eqref{pb} are represented by $\alpha$-acyclic hypergraphs.
Now, we are ready to describe a simple polynomial reduction of Max-Cut to \eqref{pb}.

\begin{proposition}\label{reduction}
Assume that an instance of Max-Cut is represented by a graph $G' = (V,E')$ and a weight vector $w \in \mathbb{Z}^{E'}_{+}$.
Then, there exists a polynomial-time reduction from Max-Cut to \eqref{pb}, where the instance of \eqref{pb} is represented by a hypergraph $G = (V,E)$ with profit vector $p \in \mathbb{Z}^{V \cup E}$ such that:
\begin{enumerate}[label=(c\arabic*)]
\item 
\label{ass: c1}
$G$ is $\alpha$-acyclic;
\item
\label{ass: c2}
all edges in $E$ have cardinality either two or $\abs{V}$;
\item
\label{ass: c3}
all edges $e \in E$ such that $\abs{E} = 2$ have profit $p_e = -2 w_e$, all edges $e \in E$ such that $\abs{E} = \abs{V}$ have profit $p_e = 0$, and all nodes $v \in V$ have profit $p_v = \sum_{u \in V \mid \{u,v\} \in E} w_{\{u,v\}}$;
\item 
\label{ass: c4}
every vector in $\{0,1\}^V$ yields the same objective value in the two problems.
\end{enumerate}
\end{proposition}

\begin{proof}
Let $I$ be an instance of Max-Cut.
We denote by $G' = (V, E')$ its associated graph, and by $w$ the weight vector for the edges in $E'$.
Let $\bar e$ be a new edge defined as $\bar e := V$.
At this point, we construct an instance $J$ of \eqref{pb}. The hypergraph representing the instance is $G = (V, E)$, where $E := E' \cup \{ \bar e \}$. 
It is easy to see that it satisfies \ref{ass: c2} by construction.
The profit vector of $J$ is $p \in \mathbb{Z}^{V \cup E}$, which is defined as
\begin{equation*}
p_i := \begin{cases}
\displaystyle\sum_{u \in V \mid \{u,v\} \in E } w_{\{u,v\}}, \quad  & \text{if } i = v \in V  \\
-2 w_{\{u,v\}}, \quad  & \text{if } i = \{u,v\} \in E'  \\
0,  & \text{if } i = \bar e \ .
\end{cases}
\end{equation*}
Clearly the vector $p$ satisfies condition \ref{ass: c3}.
Furthermore, it is immediate to see that solving $I$ is equivalent to $J$. 
In particular, the set of feasible solutions is $\{0,1\}^V$ for both Max-Cut and \eqref{pb}.
Moreover, the objective value obtained by any binary vector in $J$ coincides with the objective value yielded by the same vector in $I$.
This shows that \ref{ass: c4} holds.
It remains to prove that also \ref{ass: c1} is satisfied.
Hence, we show that $G$ is $\alpha$-acyclic.
We observe that we obtain the empty hypergraph $(\emptyset, \emptyset)$ from $G$ by first removing all edges $e \in E'$, and then by removing all nodes.
Therefore, by Theorem~\ref{charactalphaacyclic} we can conclude that the hypergraph $G$ is $\alpha$-acyclic.
\end{proof}

Next, we present the first hardness result, obtained by reducing Simple Max-Cut to \eqref{pb} using the polynomial reduction presented in Proposition~\ref{reduction}. 
Simple Max-Cut is the special case of Max-Cut, in which the weight vector $w$ is restricted to be the vector of all ones.
This problem has been shown to be strongly NP-hard in \cite{GarJohSto}.

\begin{theorem}\label{NPhard}
Solving \eqref{pb} is strongly NP-hard, even if $G = (V,E)$ is a hypergraph that satisfies conditions \ref{ass: c1}, \ref{ass: c2}, and
\begin{enumerate}
\item[(c3')] 
all edges $e \in E$ such that $\abs{e} = 2$ have profit $p_e = -2$, all edges $e \in E$ such that $\abs{e} = \abs{V}$ have profit $p_e = 0$, and all nodes $v \in V$ have profit $p_v = \abs{\{e \in E \mid v  \in e, \ \abs{e} = 2 \} }$
\end{enumerate}
\end{theorem}

Observe that condition (c3') coincides with condition \ref{ass: c3}, when we adjust the latter to Simple Max-Cut.

\smallskip 

Next, we present the hardness of approximation result.
We start by defining the concept of $r$-approximation, for any maximization problem $P$, where $r \in [0, 1]$.
Let \emph{ALG} be an algorithm that returns a feasible solution to $P$ yielding objective value \emph{ALG(I)}, for every instance $I$ of $P$.
Now, let us fix $I$.
We denote by $l(I)$ the minimum value that the objective function of $I$ can achieve on all feasible points, and by \emph{OPT(I)} the optimum value of that instance.
Then, we say that an algorithm \emph{ALG} is a \emph{$r$-approximation} for $P$ if, for every instance $I$ of $P$, we have that $\frac{ALG(I) - l(I)}{OPT(I) - l(I)} \geq r$.
In particular, when $P$ is Max-Cut, we have that $l(I) = 0$ for all instances $I$.
In \cite{TreSorSudWil} the authors show that it is NP-hard to obtain an $r$-approximation algorithm for Max-Cut, for $r > \frac{16}{17}$.
The next result then follows by reducing Max-Cut to \eqref{pb} using the reduction in Proposition~\ref{reduction}.

\begin{theorem}\label{thm: hardapprox}
It is NP-hard to obtain an $r$-approximation algorithm for \eqref{pb}, with $r > \frac{16}{17}$, even if the instance of \eqref{pb} satisfies conditions \ref{ass: c1}, \ref{ass: c2}, \ref{ass: c3}, for some vector $w \in \mathbb{Z}^E_+$.
\end{theorem}


We observe that the bound on $r$ can be further strengthened if we assume the validity of the Unique Games Conjecture, first formulated in \cite{Kho02}.
In fact, Theorem~1 in \cite{KhoKinMosODo} states 
that it is NP-hard to approximate Max-Cut to within a factor greater than $\alpha_{\text{GW}} \approx 0.878$, granted that the 
Unique Games Conjecture and the Majority Is Stablest Conjecture are true.
The constant $\alpha_{\text{GW}}$ was originally defined in \cite{GoeWil}, where the authors provide an $\alpha_{\text{GW}}$-approximation algorithm for Max-Cut.
Lastly, we observe that the Majority Is Stablest Conjecture was proved to be correct in \cite{MosODoOle}, and therefore this stronger inapproximability result now only relies on the 
Unique Games Conjecture. 

\section{Reduction scheme for general hypergraphs}
\label{sec: computation}

\begin{figure}[htbp]
\centering
\includegraphics[width=.8\textwidth]{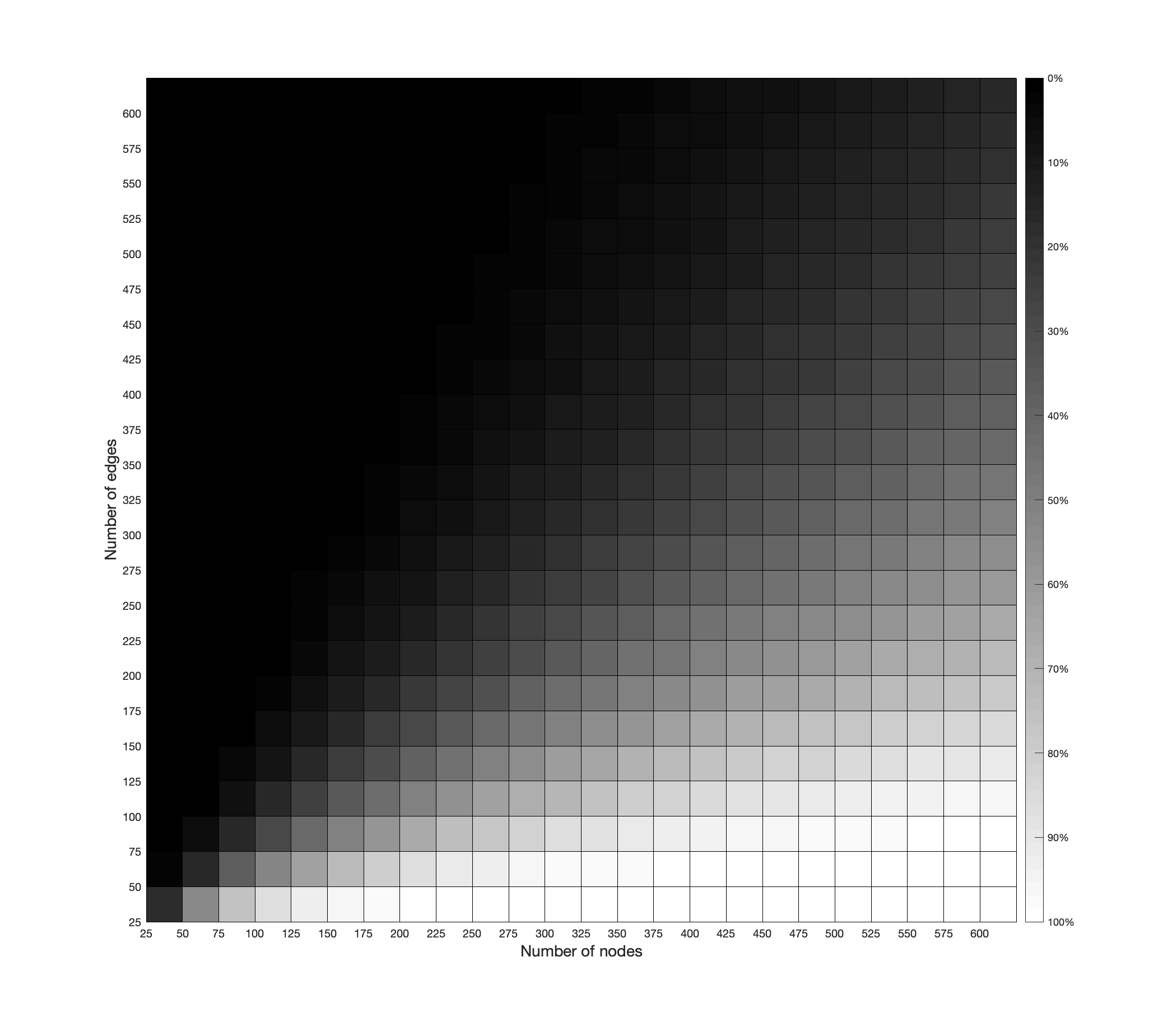}
\caption{Percentage of removed nodes in hypergraphs as a function of $\abs{V}$ and $\abs{E}$.}\label{comp_hyper}
\end{figure}

\begin{figure}[t]
\centering
\includegraphics[width=0.8\textwidth]{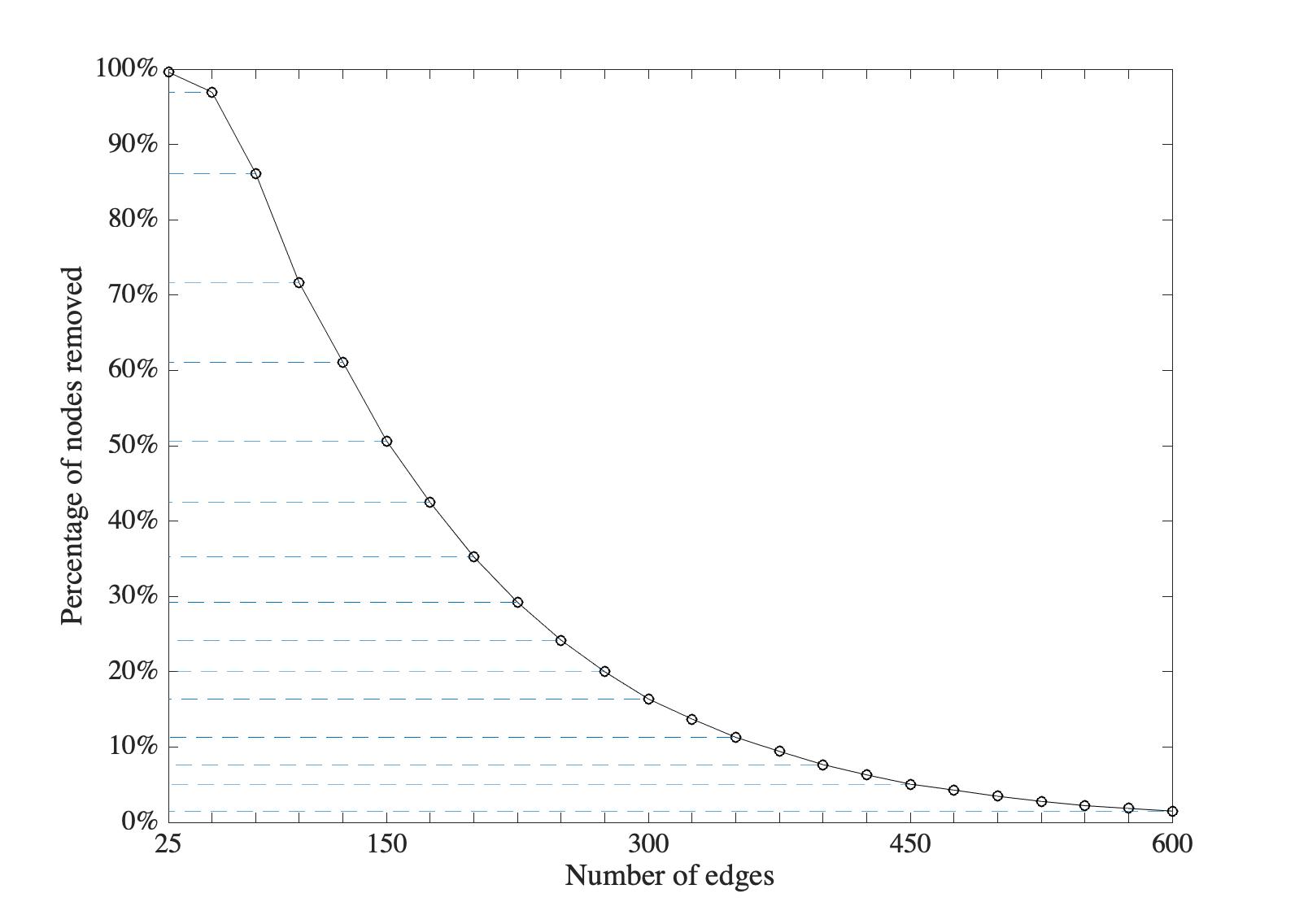} 
\caption{Percentage of removed nodes as a function of $\abs{E}$ when $n = 300$.}\label{new_slice_hyper}
\end{figure}

We observe that, even if our algorithm is not able to solve instances over hypergraphs that contain $\beta$-cycles, it is still possible to use it as a reduction scheme.
In particular, we can iteratively remove nest points, which leads to a decrease in the number of nodes, and possibly edges, of the hypergraph until there are no nest points left.
If we are able to obtain an optimal solution to the smaller problem, we can then use the rules outlined in the algorithm to construct an optimal solution to the original problem.

In order to better assess if our reduction scheme could be useful in practice, we ran some computational experiments.
We studied the reduction scheme on random instances, as it is commonly done in the literature \cite{BucRin07,CraRod17,dPKhaSah19}.
For every instance, we computed the percentage of removed nodes.
First, we explain the setting of our experiments.
We chose the setting of \cite{CraRod17}, \ie we decide the number of nodes $\abs{V}$ and of edges $\abs{E}$ of the hypergraph representing the instance, but we do not make any restriction on the rank of the hypergraph.
 We recall that the \emph{rank} of a hypergraph is the maximum cardinality of any of its edges.
For every edge, its cardinality $c$ is chosen from $\{2,\dots,\abs{V}\}$ with probability equal to $2^{1-c}$.
As explained in \cite{CraRod17}, the purpose of this choice is to model the fact that a random hypergraph is expected to have more edges of low cardinality than high cardinality.
Then, once $c$ is fixed, the nodes of the edge are chosen uniformly at random in $V$ with no repetitions.
We also make sure that there are no parallel edges in the produced hypergraph.
This will be useful in the interpretations of the results, as we explain later in the section.
The parameters $\abs{V}$ and $\abs{E}$ have values in the set $\{ 25$, $50$, $75$, \dots, $600 \}$. 
For every pair $(\abs{V},\abs{E})$ we made $250$ repetitions and computed the percentage of removed nodes.
Then, we took the average of these percentages.
The results of our simulations are shown in Figure~\ref{comp_hyper}. 
The values on the $x$ axis correspond to the number of nodes of the hypergraph, whereas the values on the $y$ axis represent the number of edges.
The lighter the cell, the more nodes are removed for instances with those values of $n$ and $m$.
A legend can be found to the right of the grid.

From the results, we noticed that the percentage of the removed nodes is related to the value of the ratio $\abs{E}/\abs{V}$, where $G = (V, E)$ is the hypergraph representing the instance.
From Figure~\ref{comp_hyper}, it is apparent that the smaller is the ratio $\abs{E}/\abs{V}$, the more effective our algorithm is.
In particular, we observe that if $\abs{E}/\abs{V} = 1$, then the average of nodes removed is 16.72\%.
However, when $\abs{E}/\abs{V} =  1/2$, this percentage is roughly 50\%, and if $\abs{E}/\abs{V} =  1/4$ our algorithm removes on average 86\% of the original nodes.
Additional values can be extracted from Figure~\ref{new_slice_hyper}, which captures the trend of this percentage as a function of $\abs{E}$.
In this figure, the number of nodes $\abs{V}$ is set to 300. 
We see that the reduction scheme is particularly useful whenever $\abs{E}/\abs{V} \leq 1$, \ie when the number of edges is bounded by the number of nodes.
Furthermore, we observe that a large subset of the hypergraphs with $\abs{E}/\abs{V} \leq 1$ have a highly non-trivial structure, since they have a huge connected component with high probability. 
In fact, the largest connected component of $G$ is of order $\abs{V}$ whenever the fraction $\abs{E}/\abs{V}$ is asymptotic to a constant $c$ such that $c > 1/2$.
This follows from \cite{ErdRen} once we observe that each edge of a hypergraph connects at least as many nodes as an edge in a graph.
We remark that the authors in \cite{ErdRen} do not allow parallel edges, and this is why we introduced this requirement for our instances.

For denser hypergraphs, i.e., hypergraphs with $\abs{E}/\abs{V} > 1$, our procedure does not work as well, and this can be explained by the fact that, for these hypergraphs, it is more unlikely that a node would be able to satisfy the definition of nest point.
For non-random instances, it should be noted that the outcome of our reduction scheme depends heavily on the structure of the specific instance.

Lastly, we remark that the reduction scheme can be applied also to quadratic instances, where the corresponding hypergraph is simply a graph, and where nest points are leaves.
We wanted to check if the computational experiments would lead to comparable findings when $G$ is actually a graph.
This is indeed the case, as Figure~\ref{computations graphs} indicates.
\begin{figure}[htbp]
\centering
\label{computations}
\includegraphics[width=.8\textwidth]{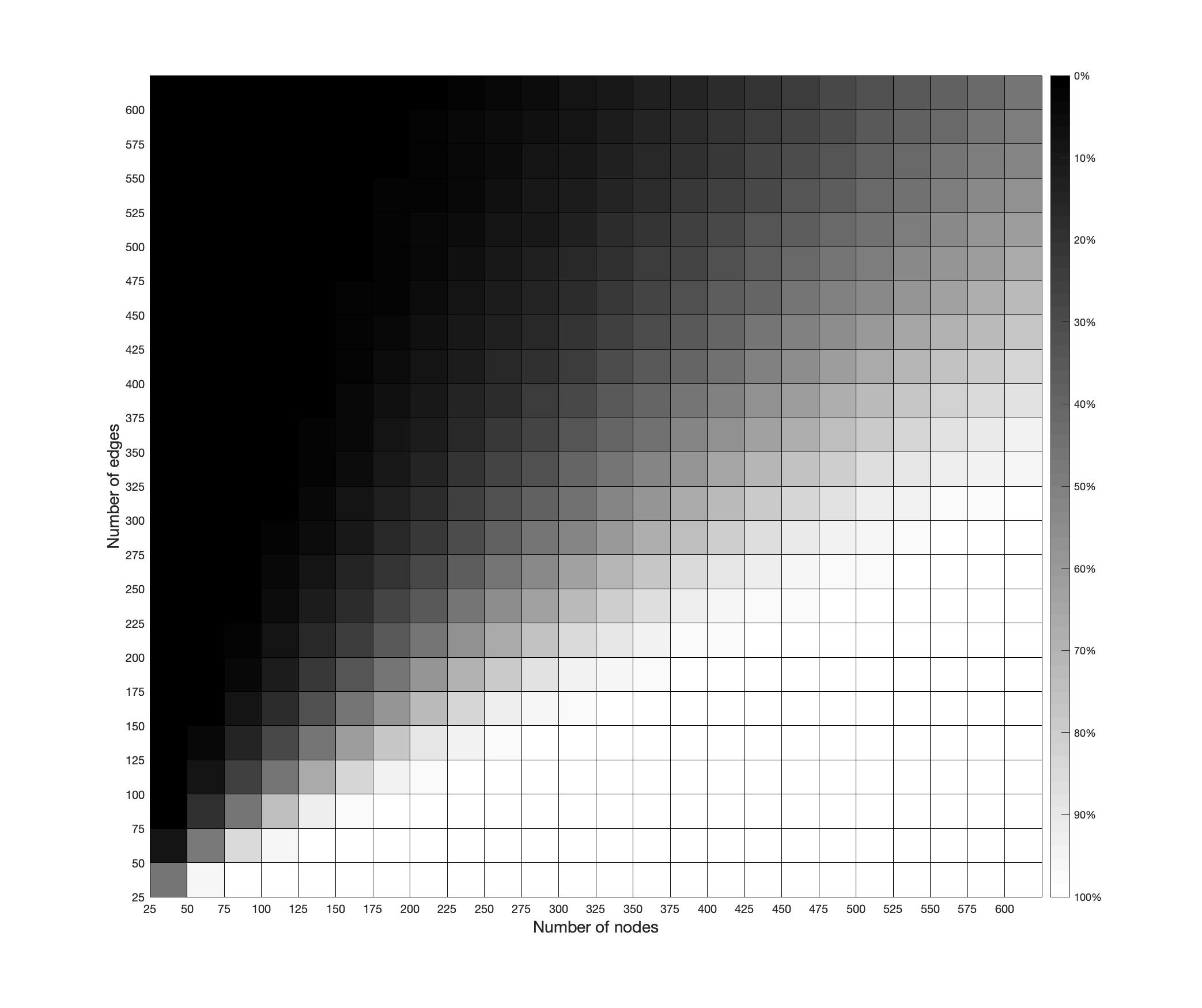}
\caption{Percentage of removed nodes in graphs as a function of $\abs{V}$ and $\abs{E}$.}\label{computations graphs}
\end{figure}
In fact, the behavior of the percentages of removed nodes is similar to the one in the hypergraph setting, even if the shift between the light and dark regions is sharper.
In order to unveil better this more radical performance, we look again at the average of the percentages of nodes removed as a function of the ratio $\abs{E}/\abs{V}$.
In particular, we look at the same values of this ratio that we explicitly mentioned in the hypergraph setting, that is for $\abs{E}/\abs{V} \in \left\{1, \frac12, \frac14 \right\}$.
From the computational experiments we see that these values are respectively $45.63\%$, $97.56\%$, and $99.88\%$. 
To further highlight this behavior in the graph setting, we fix $\abs{V} = 300$ and study the average of the percentage of removed nodes as a function of $\abs{E}$.
From this analysis, the reader can derive additional values corresponding to the ratio $\abs{E}/\abs{V}$ from Figure~\ref{slice_graph}.  
\begin{figure}[H]
\centering
\includegraphics[width=\textwidth]{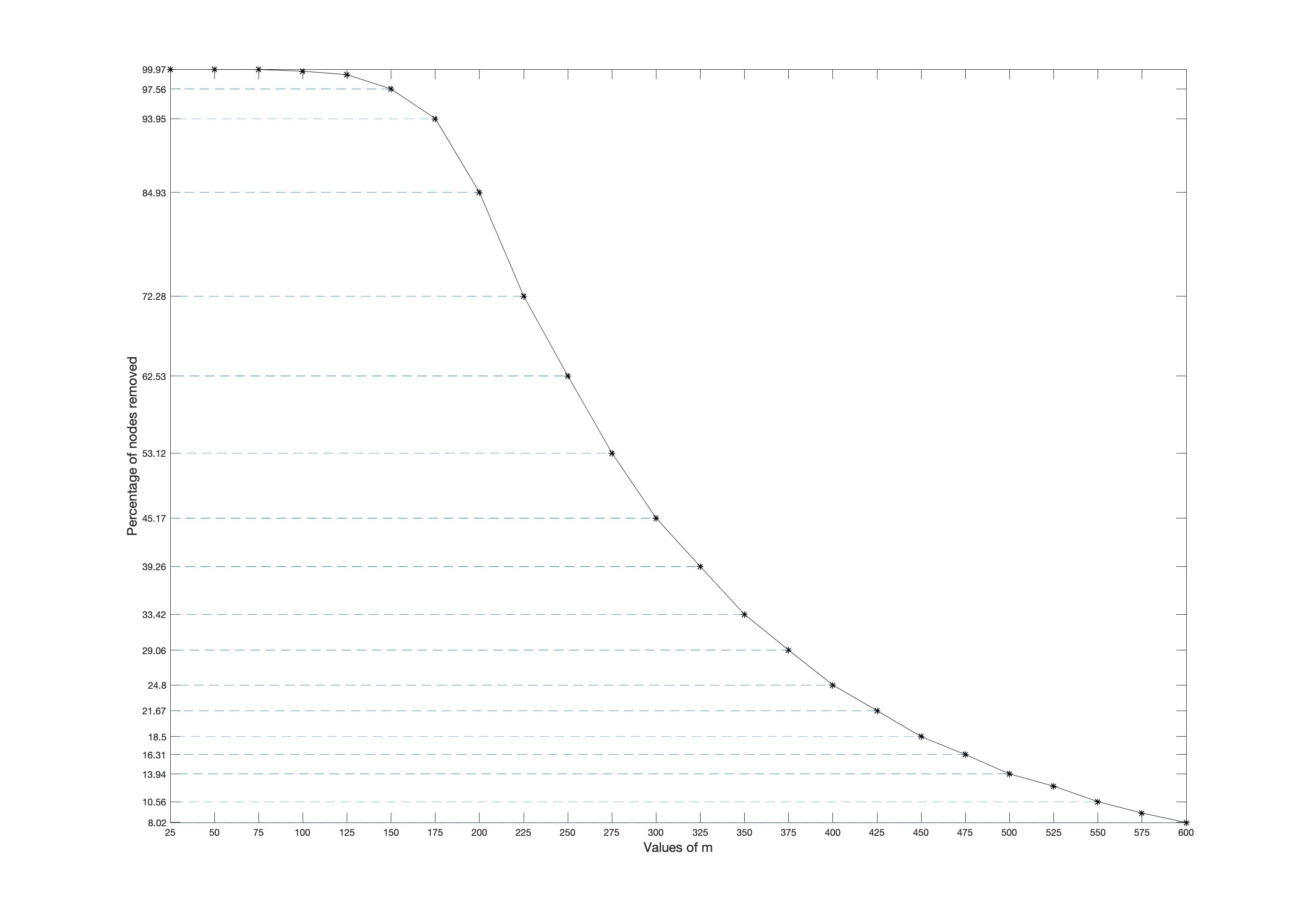}
\caption{Percentage of removed nodes in graphs as a function of $\abs{E}$ when $n = 300$.}\label{slice_graph}
\end{figure}

\bigskip
\noindent
\textbf{Funding:} A.~Del~Pia is partially funded by ONR grant N00014-19-1-2322. 
S.~Di~Gregorio is partially supported by NSF award CMMI-1634768.
Any opinions, findings, and conclusions or recommendations expressed in this material are those of the authors and do not necessarily reflect the views of ONR or NSF.

%
%

\bibliographystyle{plainurl}


\begin{thebibliography}{10}

\bibitem{Bar86}
F.~Barahona.
\newblock A solvable case of quadratic $0-1$ programming.
\newblock {\em Discrete Applied Mathematics}, 13(1):23--26, 1986.

\bibitem{BeFaMaYa83}
C.~Beeri, R.~Fagin, D.~Maier, and M.~Yannakakis.
\newblock On the desirability of acyclic database schemes.
\newblock {\em Journal of the ACM}, 30:479--513, 1983.

\bibitem{Beietal17}
T.~Beier, C.~Pape, N.~Rahaman, T.~Prange, S.~Berg, D.~Bock, A.~Cardona,
  G.~Knott, S.~Plaza, L.~Scheffer, U.~Koethe, A.~Kreshuk, and F.~Hamprecht.
\newblock Multicut brings automated neurite segmentation closer to human
  performance.
\newblock {\em Nature Methods}, 14(2):101--102, 2017.

\bibitem{BerTsiBook}
D.~Bertsimas and J.N. Tsitsiklis.
\newblock {\em Introduction to Linear Optimization}.
\newblock Athena Scientific, Nashua NH, 1997.

\bibitem{BieMun18}
D.~Bienstock and G.~Mu{\~n}oz.
\newblock {LP} formulations for polynomial optimization problems.
\newblock {\em SIAM Journal on Optimization}, 28(2):1121--1150, 2018.

\bibitem{BorCraRod20}
E.~Boros, Y.~Crama, and E.~Rodr\'iguez-Heck.
\newblock Compact quadratizations for pseudo-boolean functions.
\newblock {\em Journal of Combinatorial Optimization}, 39:687--707, 2020.

\bibitem{BorGru12}
E.~Boros and A.~Gruber.
\newblock On quadratization of pseudo-boolean functions.
\newblock {\em Preprint, arXiv:1404.6538.}, 2012.

\bibitem{BorHam02}
E.~Boros and P.L. Hammer.
\newblock Pseudo-boolean optimization.
\newblock {\em Discrete Applied Mathematics}, 123:155--225, 2002.

\bibitem{Bra14}
J.~Brault-Baron.
\newblock Hypergraph acyclicity revisited.
\newblock {\em ACM Computing Surveys}, 49(3):54:1--54:26, 2016.

\bibitem{BucCraRod18}
C.~Buchheim, Y.~Crama, and E.~Rodr\'iguez-Heck.
\newblock Berge-acyclic multilinear $0-1$ optimization problems.
\newblock {\em European Journal of Operational Research}, 273(1):102--107,
  2018.

\bibitem{BucRin07}
C.~Buchheim and G.~Rinaldi.
\newblock Efficient reduction of polynomial zero-one optimization to the
  quadratic case.
\newblock {\em SIAM Journal on Optimization}, 18(4):1398--1413, 2007.

\bibitem{ConCorZamBook}
M.~Conforti, G.~Cornu\'ejols, and G.~Zambelli.
\newblock {\em Integer Programming}.
\newblock Springer Publishing Company, Incorporated, 2014.

\bibitem{Cra93}
Y.~Crama.
\newblock Concave extensions for non-linear $0-1$ maximization problems.
\newblock {\em Mathematical Programming}, 61(1):53--60, 1993.

\bibitem{CraHanJau90}
Y.~Crama, P.~Hansen, and B.~Jaumard.
\newblock The basic algorithm for pseudo-boolean programming revisited.
\newblock {\em Discrete Applied Mathematics}, 29:171--185, 1990.

\bibitem{CraRod17}
Y.~Crama and E.~Rodr\'iguez-Heck.
\newblock A class of valid inequalities for multilinear $0-1$ optimization
  problems.
\newblock {\em Discrete Optimization}, 25:28--47, 2017.

\bibitem{dPDiG21}
A.~Del~Pia and S.~Di~Gregorio.
\newblock Chv\'atal rank in binary polynomial optimization.
\newblock {\em INFORMS Journal on Optimization}, 2021.
\newblock \href {https://doi.org/10.1287/ijoo.2019.0049}
  {\path{doi:10.1287/ijoo.2019.0049}}.

\bibitem{dPKha17}
A.~Del~Pia and A.~Khajavirad.
\newblock A polyhedral study of binary polynomial programs.
\newblock {\em Mathematics of Operations Research}, 42(2):389--410, 2017.

\bibitem{dPKha18SIOPT}
A.~Del~Pia and A.~Khajavirad.
\newblock The multilinear polytope for acyclic hypergraphs.
\newblock {\em SIAM Journal on Optimization}, 28(2):1049--1076, 2018.

\bibitem{dPKha18MPA}
A.~Del~Pia and A.~Khajavirad.
\newblock On decomposability of multilinear sets.
\newblock {\em Mathematical Programming, Series A}, 170(2):387--415, 2018.

\bibitem{dPKha21}
A.~Del~Pia and A.~Khajavirad.
\newblock The running intersection relaxation of the multilinear polytope.
\newblock {\em Mathematics of Operations Research}, 46(3):1008--1037, 2021.

\bibitem{dPKhaSah19}
A.~Del~Pia, A.~Khajavirad, and N.~Sahinidis.
\newblock On the impact of running-intersection inequalities for globally
  solving polynomial optimization problems.
\newblock {\em Mathematical Programming Computation}, 12:165--191, 2020.

\bibitem{Dur12}
D.~Duris.
\newblock Some characterizations of $\gamma$ and $\beta$-acyclicity of
  hypergraphs.
\newblock {\em Information Processing Letters}, 112:617--620, 2012.

\bibitem{ErdRen}
P.~Erd{\H o}s and A.~R{\'e}nyi.
\newblock On the evolution of random graphs.
\newblock In {\em Publication Of The Mathematical Institute Of The Hungarian
  Academy Of Sciences}, pages 17--61, 1960.

\bibitem{Fag83Painless}
R.~Fagin.
\newblock Acyclic database schemes (of various degrees): A painless
  introduction.
\newblock In {\em CAAP}, 1983.

\bibitem{Fag83}
R.~Fagin.
\newblock Degrees of acyclicity for hypergraphs and relational database
  schemes.
\newblock {\em Journal of the Association for Computing Machinery},
  30(3):514--550, 1983.

\bibitem{For}
R.~Fortet.
\newblock Applications de l'alg\'ebre de boole en recherche op\`erationelle.
\newblock {\em Revue Fran{\c c}aise d'Automatique, Informatique et Recherche
  Op{\'e}rationnelle}, 4:17--26, 1960.

\bibitem{FreDri05}
D.~Freedman and P.~Drineas.
\newblock Energy minimization via graph cuts: Settling what is possible.
\newblock {\em CVPR, IEEE Computer Society}, pages 939--946, 2005.

\bibitem{GarJohSto}
M.R. Garey, D.S. Johnson, and L.~Stockmeyer.
\newblock Some simplified np-complete graph problems.
\newblock {\em Theoretical Computer Science}, pages 237--267, 1976.

\bibitem{GoeWil}
M.~X. Goemans and D.~P. Williamson.
\newblock Improved approximation algorithms for maximum cut and satisfiability
  problems using semidefinite programming.
\newblock {\em Journal of the ACM}, 42:1115--1154, 1995.

\bibitem{HamRosRud63}
P.L. Hammer, I.~Rosenberg, and S.~Rudeanu.
\newblock On the determination of the minima of pseudo-boolean functions (in
  romanian).
\newblock {\em Studii si Cercetari Matematice}, 14:359--364, 1963.

\bibitem{HamRud68}
P.L. Hammer and S.~Rudeanu.
\newblock {\em {B}oolean methods in operations research and related areas}.
\newblock Springer, Berlin, New York, 1968.

\bibitem{HojPfeWal19}
C.~Hojny, M.E. Pfetsch, and M.~Walter.
\newblock Integrality of linearizations of polynomials over binary variables
  using additional monomials.
\newblock {\em Manuscript}, 2019.
\newblock URL:
  \url{http://www.optimization-online.org/DB_FILE/2019/11/7481.pdf}.

\bibitem{Ish09}
H.~Ishikawa.
\newblock Higher-order gradient descent by fusion-move graph cut.
\newblock {\em ICCV, IEEE}, pages 568--574, 2009.

\bibitem{Ish11}
H.~Ishikawa.
\newblock Transformation of general binary mrf minimization to the first-order
  case.
\newblock {\em IEEE Trans. Pattern Anal. Mach. Intell.}, 33(6):1234--1249,
  2011.

\bibitem{JegNdo09}
P.~J\'egoua and S.N. Ndiayeb.
\newblock On the notion of cycles in hypergraphs.
\newblock {\em Discrete Mathematics}, 309:6535--6543, 2009.

\bibitem{Keu17}
M.~Keuper.
\newblock Higher-order minimum cost lifted multicuts for motion segmentation.
\newblock In {\em 2017 IEEE International Conference on Computer Vision
  (ICCV)}, pages 4252--4260, 2017.

\bibitem{Kho02}
S.~Khot.
\newblock On the power of unique 2-prover 1-round games.
\newblock {\em Proceedings of the 34th ACM Symposium on Theory of Computing},
  pages 767--775, 2002.

\bibitem{KhoKinMosODo}
S.~Khot, G.~Kindler, E.~Mossel, and R.~O'Donnell.
\newblock Optimal inapproximability results for max-cut and other two-variable
  csps?
\newblock {\em Proceedings of the 45th IEEE Symposium on Foundations of
  Computer Science}, pages 146--154, 2004.

\bibitem{KocHaoGloLewLuWanWan14}
G.~Kochenberger, J.-K. Hao, F.~Glover, M.~Lewis, Z.~L\"u, H.~Wang, and Y.~Wang.
\newblock The unconstrained binary quadratic programming problem: a survey.
\newblock {\em Journal of Combinatorial Optimization}, 28:58--81, 2014.

\bibitem{LanAnd21}
J.-H. Lange and B.~Andres.
\newblock On the lifted multicut polytope for trees.
\newblock In {\em Pattern Recognition}, volume 12544, pages 360--372. 42nd DAGM
  German Conference, DAGM GCPR 2020, 2021.

\bibitem{Lau09}
M.~Laurent.
\newblock Sums of squares, moment matrices and optimization over polynomials.
\newblock In {\em Emerging Applications of Algebraic Geometry}, volume 149 of
  {\em The IMA Volumes in Mathematics and its Applications}, pages 157--270.
  Springer, New York, NY, 2009.

\bibitem{Mic21}
C.~Michini.
\newblock Tight cycle relaxations for the cut polytope.
\newblock {\em To appear in SIAM Journal on Discrete Mathematics}, 2021.

\bibitem{MosODoOle}
E.~Mossel, R.~O'Donnell, and K.~Oleszkiewicz.
\newblock Noise stability of functions with low influences: invariance and
  optimality.
\newblock {\em Proceedings of the 46th IEEE Symposium on Foundations of
  Computer Science}, pages 21--30, 2005.

\bibitem{OrdPauSze13}
S.~Ordyniak, D.~Paulusma, and S.~Szeider.
\newblock Satisfiability of acyclic and almost acyclic cnf formulas.
\newblock {\em Theoretical Computer Science}, 481:85--99, 2013.

\bibitem{Pad89}
M.~Padberg.
\newblock The {B}oolean quadric polytope: Some characteristics, facets and
  relatives.
\newblock {\em Mathematical Programming}, 45(1--3):139--172, 1989.

\bibitem{Ros75}
I.~Rosenberg.
\newblock Reduction of bivalent maximization to the quadratic case.
\newblock {\em Cahiers Centre {\'E}tudes Recherche Op{\'e}r.}, 17:71--74, 1975.

\bibitem{SchBookIP}
A.~Schrijver.
\newblock {\em Theory of Linear and Integer Programming}.
\newblock Wiley, Chichester, 1986.

\bibitem{SchBookCO}
A.~Schrijver.
\newblock {\em Combinatorial Optimization. Polyhedra and Efficiency}.
\newblock Springer-Verlag, Berlin, 2003.

\bibitem{Tangetal17}
S.~Tang, M.~Andriluka, B.~Andres, and B.~Schiele.
\newblock Multiple people tracking by lifted multicut and person
  re-identification.
\newblock In {\em 2017 IEEE Conference on Computer Vision and Pattern
  Recognition (CVPR)}, pages 3701--3710, 2017.

\bibitem{TreSorSudWil}
L.~Trevisan, G.~B. Sorkin, M.~Sudan, and D.~P. Williamson.
\newblock Gadgets, approximation, and linear programming.
\newblock {\em SIAM Journal on Computing}, 29(6):2074--2097, 2000.

\end{thebibliography}

%
%

\end{document}